\documentclass[11pts]{article}

\usepackage{stfloats,amsmath,epsfig}

\usepackage{url}

\usepackage{amsfonts}

\newtheorem{definition}{Definition}

\newtheorem{proposition}{Proposition}
\newtheorem{theorem}{Theorem}

\newenvironment{proof}{{\sc Proof:}}{ \indent}

\def\bbN{{\mathbb N}}
\def\bbZ{{\mathbb Z}}
\def\bbR{{\mathbb R}}
\def\bbC{{\mathbb C}}
\def\bbI{{\mathbb I}}

\def\P{{\mathbb P}}
\def\E{{\mathbb E}}

\def\what{\widehat}
\def\noi{\noindent}
\def\MSE{{\rm m.s.e.}}

\def\refeq#1{(\ref{e:#1})}

\begin{document}
%
%
%
\title{Global Modeling and Prediction of Computer Network Traffic}

\author{Stilian A. Stoev, George Michailidis, and 
Joel Vaughan}

%

\date{Dec 17, 2009 \\ {\it Department of Statistics, the University of Michigan}}

\maketitle

\begin{abstract}
 We develop a probabilistic framework for {\em global} modeling of the traffic over
 a computer network.  This model integrates existing single--link (--flow)
 traffic models with the routing over the network to capture the 
 global traffic behavior.  It arises from a limit approximation of the
 traffic fluctuations as the time--scale and the number of users sharing
 the network grow.  The resulting probability model is comprised of
 a Gaussian and/or a stable, infinite variance components.  They can be 
 succinctly described and handled by certain 'space--time' random fields.
 The model is validated against simulated and real data.  It is then 
 applied to predict traffic fluctuations over unobserved links from a 
 limited set of observed links.  Further, applications to anomaly detection and network
 management are briefly discussed.
\end{abstract}

%
%

\section{Introduction}

 Understanding the statistical behavior of computer network traffic
 has been an important and challenging problem for the past 15 years, because of its
 impact on network performance and provisioning \cite{lombardo:morabito:schembra:2004,
 rolls:campos:michailidis:2005,gong:liu:misra:towsley:2005,park:willinger:2000}
 and on the potential for development of more suitable protocols \cite{park:kim:crovella:1996R,park:willinger:2000}.
 Since the early 1990s it has been well established that the traffic over a {\em single link} exhibits intricate 
 temporal dependence, known as {\it burstiness}, which could not be explained by traffic models
 developed for telephone networks \cite{leland:taqqu:willinger:wilson:1993Sig}. This phenomenon could be 
 understood and described by using the notions of long--range dependence and self--similarity 
 \cite{erramilli:pruthi:willinger:1995S}, which in turn are affected by the presence of heavy tails
 in the distribution of file sizes  \cite{crovella:bestavros:1996,park:kim:crovella:1996R}.
 A bottom-up {\em mechanistic} model for single link network traffic that is in agreement with the empirical
 features observed in real network traces was presented in \cite{taqqu:willinger:sherman:1997}. A competing
 model based on queuing ideas was studied in \cite{mikosch:resnick:rootzen:stegeman:2002}. These works lead to
 many further developments (see eg \cite{park:willinger:2000}).
 
Advances in technology that allowed the acquisition of direct, through sampling \cite{duffield:2004,yang:michailidis:2007},
and indirect \cite{lawrence:etal:2006} measurements have allowed researchers to examine the
characteristics of traffic in entire networks \cite{lakhina:etal:2004, gong:liu:misra:towsley:2005,singhal:michailidis:2007,
xu:zhang:bhattacharyya:2005}, based on statistical modeling analysis. On the other hand, an analogue
of the mechanistic models available for single link network traffic is not available. Such a model would allow
better understanding of network performance \cite{field:harder:harrison:2004,lombardo:morabito:schembra:2004} and detection of anomalous
behavior \cite{paschalidis:smaragdakis:2009}. Further, it would manage to capture and explain statistical relationships
between flows traversing the network at all time scales (time) and across all links (space); the latter represents
a fairly tall requirement, which may also prove rather impractical given the underlying complexity (protocols, applications)
and heterogeneity (physical infrastructure, diverse users) of modern networks. 

Our objective in this paper is to propose a mechanistic model that captures several fundamental
characteristics of network-wide traffic
and thus constitutes a partial solution for this challenging problem. The model is based on modeling user behavior on
source--destination paths across the network and then aggregate over users and over time, thus developing a 
{\em joint} 'space--time' probability model for the traffic fluctuations over {\em all links} in the network.  This model 
reflects the statistical dependence of the traffic across different links, observed at the same or different points in time.
We demonstrate the success of our modeling strategy in the context of {\em network traffic prediction} -- a problem with 
important implications on network performance, provisioning, and management.

{\em The remainder of the paper is structured as follows.} In Section \ref{s:single-link}, we review 
briefly the existing and relatively well--understood theory of single--flow (link) models for the temporal 
dependence in network traffic.  Long--range dependence and heavy tails play a central role.  In Section 
\ref{s:network-wide}, we postulate our network--wide model based on combining single--flow models through
the routing equation.  We show that the scaling limit of such a model is a combination of fractional Brownian 
motions and infinite variance stable L\'evy motions.  A succinct representation of these processes is given in
Section \ref{s:f-fBm,f-Lsm} via the {\em functional fractional Brownian motion} and {\em functional L\'evy stable 
motion}. The resulting model is then used in Section \ref{s:kriging} to solve the {\em network kriging} problem, 
i.e.\ to predict the traffic fluctuations on a unobserved link from a limited set of measurements of observed links.
In Section \ref{s:NetFlow}, we use extensive NetFlow data of sampled network traffic to obtain approximations of
the flow--level traffic $X_j(t)$.  These data are then used to validate our model and demonstrate the success of 
the network kriging methodology.  We conclude in Section \ref{s:discussion} with some remarks on future applications,
statistical problems on networks, and further extensions of the network--wide probabilistic model.

\section{Problem Formulation}
\label{problem-formulation}

Consider a computer network of $L$ links and $N$ nodes.  The network typically 
carries traffic {\em flows} (via groups of packets) from any node (source) to any other node (destination) over a 
predetermined set of links (route).  This can be formally described by the {\em routing matrix}
$A = (a_{\ell j})_{L \times {\cal J}}$, where
$$
 a_{\ell j} = \left\{ \begin{array}{ll}
 1 &,\ \mbox{ {\em route} $j$ involves {\em link} $\ell$ }\\
 0 &,\ \mbox{ otherwise,}
 \end{array}\right.
$$
and where ${\cal J}$ is the total number of routes (typically, ${\cal J} =N(N-1)$).

{\em We describe next the physical premises of our modeling framework.}
We assume, for simplicity, that the traffic is {\em fluid}. That is, the amount of data (bytes)
transmitted over link $\ell$ during the time interval $(a,b)$ is $\int_a^b Y_\ell(t) dt$, where
$Y_\ell(t)$ is the {\em traffic intensity} (bytes per unit time) over link $\ell$.  Let also 
$X_j(t)$ denote the traffic intensity at time $t$ over route $j,\ 1\le j\le {\cal J}$.  Then,
assuming that traffic propagates instantaneously over the network, we obtain the following 
{\em routing equation}:
\begin{equation}\label{e:Y=AX}
 \vec Y(t) = A \vec X(t),
\end{equation}
where $\vec X(t) = (X_j(t))_{1\le j\le {\cal J}}$ and $\vec Y(t) = (Y_\ell(t))_{1\le \ell \le L}$. 
This relationship is valid only to the extent that traffic propagates {\em instantaneously} along
the routes. Thus, \eqref{e:Y=AX} cannot be adopted over the finest, high--frequency time scales where
packet delay plays a central role. On the other hand, for all practical purposes, the routing equation
holds over a wide range of time scales greater than the RTT (round trip time) for packets in the network
\cite{singhal:michailidis:2007,yang:michailidis:2007}.  This equation captures the fundamental
relationship between the traffic intensity over different routes in the network and the resulting load, 
incurred on the links.

From a physical perspective, the computer network is merely used to 'transport' information from source nodes to 
destination nodes.  In normal (uncongested) operating regime, the traffic is carried seamlessly and 
the traffic intensities $X_j(t)$ are driven solely by the {\em demand} along the routes $j,\ 1\le j \le {\cal J}.$  
Thus, as a first approximation one may view the $X_j(t)$'s as statistically independent in $j,\ 1\le j\le {\cal J}$.
Therefore, in view of \eqref{e:Y=AX}, the statistical dependence between $Y_{\ell_1}(t)$ and $Y_{\ell_2}(t)$ for two
links $\ell_1$ and $\ell_2$, is governed by the set of routes $X_j(t)$ that use both $\ell_1$ and $\ell_2$.

In view of the above discussion, guided by the routing equation \eqref{e:Y=AX}, we obtain a global model
for the traffic intensity $Y_\ell(t),\ 1\le \ell \le L,\ t\ge 0$.  The temporal dependence of the flow--level traffic 
$X_j(t)$ can be described well by the existing mechanistic models exhibiting long--range dependence and heavy tails
(see Section \ref{s:single-link}).  The independence of the $X_j(t)$'s in $j$ is a questionable assumption when
the network is not in equilibrium or it is congested.  Indeed, if two routes share a congested node, then the
feedback mechanism of TCP clearly induces dependence between the two flows.  Further, since every TCP session
involves ACK (acknowledgment) packets traversing along the {\em reverse} route, then in practice one expects 
dependence between the {\em forward} and {\em reverse} flows for a given pair of a source and destination. 
Our experience with NetFlow data for the Internet2 backbone network suggests however that for the present
utilization levels of the network (about 10\% to 20\%) the $X_j(t)$'s are nearly uncorrelated in $j$ (see Fig.\ 
\ref{fig:X-corr} in Section \ref{s:NetFlow}).  The correlation is strongest but still rather weak among the forward
and reverse flows (see also \cite{singhal:michailidis:2007}).

Therefore, as a first attempt to model globally the dependence structure of the network across all links and
in time, we advocate adopting the simple assumption of independence of the flow--level traffic.  Our methodology
can be extended to cover more complex scenarios of dependence between forward and reverse flows, as well as 
'second--order' effects of dependence between flows triggered by congestion.  This should be done with caution 
however since the chaotic behavior induced by the TCP feedback is not well--understood on network--wide level.

\subsection{A Brief Overview of Single Link Traffic Models}
\label{s:single-link}

We start with a brief review of single--link traffic models, since a number of their
features are incorporated into our network--wide model.
 Such models are built on the paradigm of multiple users
 sharing a link.  Depending on the regimes prevalent in the network,
 one obtains two qualitatively different asymptotic models for the
 cumulative traffic fluctuations. One regime leads to
 finite--variance, Gaussian models that exhibit long--range
 dependence and self--similarity.  The other regime yields infinite
 variance processes with independent increments. 

 \subsection{Activity rate models: two limit regimes}
 
 Consider a fixed route on the network and suppose that $M$ independent users share 
 this route.   Let $\{X(t)\}_{t\ge 0}$ denote the traffic intensity of one such user
 in bytes per unit time. Thus $\int_a^b X(t) dt$ is the total traffic (bytes)
 generated by the user during the time interval $(a,b)$. It is assumed that $\{X(t)\}_{t\ge 0}$ 
 is a strictly stationary stochastic process with finite mean.

 Following the framework in 
\cite{mikosch:samorodnitsky:2007}, let $\{ (T_j,Z_j)\}_{j\in \bbZ}$ be
a stationary marked point process of arrival times $T_j$'s in $\bbR$
with marks $Z_j$'s.  At time $T_j$, the user initiates a transmission
at constant unit rate, which lasts time $Z_j$.  Thus, the traffic
intensity at time $t$ equals:
\begin{equation}\label{e:act-rate}
  X(t) = \sum_{j\in\bbZ} I(T_j\le t< T_j + Z_j),
\end{equation}
where $ \cdots \le T_0 \le 0 \le T_1 \le \cdots$. One can recover the following two popular traffic
models as special cases:

\medskip
\begin{itemize}
\item {\em $M/G/\infty$ model:} If the $T_j$'s are arrival times of a Poisson point process with constant 
intensity, independent of the marks $Z_j$'s, then $\{X(t)\}_{t\ge 0}$ becomes the $M/G/\infty$ model.  

\item {\em On/Off model:} On the other hand, if the $Z_j$'s and the $T_j$'s are {\em dependent} and such that:
$$
 T_{j,{\rm on}} := T_j,\ \ T_{j,{\rm off}} := T_j + Z_j < T_{j+1} \equiv T_{j+1,{\rm on}},
$$
then $X(t)$ follows the {\em On/Off model}, i.e.\ a period of activity ('On') is followed by an idle period ('Off').
It is further assumed that the {\em On} periods: $U_{j,{\rm on}} := Z_j$ and the {\em Off} periods 
$U_{j,{\rm off}}:= T_{j+1} - T_j + Z_j$, are mutually independent and identically distributed with laws 
$F_{\rm on}(x) = \P\{ U_{1,{\rm on}} \le x\}$ and $F_{\rm off} (x) = \P\{U_{1, {\rm off}} \le x\}$.

\end{itemize}

%
%
\noi{\em Remark.} The initial {\em On} period $T_{1,\rm on}$ has such a distribution as
to ensure that $\{X(t)\}$ is stationary.  This can happen only if the {\em On} and {\em Off}
periods have finite means.  The work \cite{mikosch:resnick:2006} addresses the case
of activity rates with very heavy tails, which can have infinite means.

\medskip
In the context of network traffic, the {\em durations} of the user
activity $Z_j$'s are modeled with heavy tailed distributions with
finite mean but infinite variance, since they can be linked to the ubiquitous presence of heavy
tails in computer networks (file sizes, web pages, etc.\ see e.g.\
\cite{crovella:bestavros:1996,park:kim:crovella:1996R,crovella:taqqu:bestavros:1998}).
The heavy tailed nature of the durations, implies that the process
$X(t)$ of user activity is long--range dependent (LRD).  The intimate
connection between the long--range dependence phenomenon and
self--similarity provides an appealing mechanistic (physical)
explanation of the cause of burstiness in network traffic (see e.g.\
\cite{leland:taqqu:willinger:wilson:1993Sig,erramilli:pruthi:willinger:1995S,willinger:taqqu:sherman:wilson:1997TON},
\cite{park:willinger:2000} and the references therein).

For brevity, we focus on the {\em On/Off model} and suppose that the tails of the On and Off durations 
are heavy, i.e.\ as $x\to\infty$:
$$
1 - F_{\rm on}(x):= \overline{F}_{\rm on}(x) \sim c_{\rm on} x^{-\alpha_{\rm on}}\
 \mbox{ and }\ \overline{F}_{\rm Off}(x) = c_{\rm off} x^{-\alpha_{\rm off}},
$$ 
for some constants $c_{\rm on},\ c_{\rm off}>0$ and tail exponents, such that
\begin{equation}\label{e:alphas}
   1 < \alpha:= \min\{ \alpha_{\rm on}, \alpha_{\rm off}\} <2. 
\end{equation}
Relation \eqref{e:alphas} then implies that $X(t)$ is LRD with Hurst exponent
\begin{equation}\label{e:H}
 H = \frac{3 - \alpha}{2} \in (1/2,1),
\end{equation}
that is, for some constant $c_X>0$,
$$
{\rm Cov}(X(t), X(0)) \sim c_X t^{2H - 2},\ \ \mbox{ as }t\to\infty
$$
(see e.g.\ \cite{mikosch:samorodnitsky:2007}).

 \subsection{Multiple Sources Asymptotics: Long--Range Dependence and Heavy Tails}

 Let now $\{X^{(i)}(t)\}$,\ $1\le i\le M$ be independent and identically distributed stationary
 processes modeling the traffic intensities of $M$ users sharing a given route.  Then, the 
 cumulative traffic over the route generated by the users is:
 $$
  X^*(T,M):= \int_0^T \sum_{i=1}^M X^{(i)}(t) dt.
 $$
 We are interested in the asymptotic behavior of the cumulative traffic fluctuations about the mean:
$$
 X_0^*(T,M) :=   X^*(T,M) - \E X^{*}(T,M).
$$

 As shown in the seminal work of \cite{taqqu:willinger:sherman:1997},
 if the $X^{(i)}(t)$'s are {\em On/Off processes}, then 
\begin{equation}\label{e:fast-regime}
 {\cal L}\lim_{T\to\infty} \frac{1}{T^H} {\Big\{} {\cal L}\lim_{M\to\infty}
 \frac{1}{\sqrt{M}} X_0^*(Tt,M) {\Big\}}_{t\ge 0} = \{B_H(t)\}_{t\ge 0},
\end{equation}
where $B_H = \{B_H(t)\}_{t\ge 0}$ is a fractional Brownian motion (fBm) with self--similarity parameter $H$ as in 
\eqref{e:H} and where '${\cal L}\lim$' denotes finite--dimensional distributions convergence. 
Recall that the fBm $B_H$ is a zero mean Gaussian process with stationary increments, which is self--similar, i.e.\ 
for all $c>0$, we have $\{B_H(ct)\}_{t\ge 0} \stackrel{d}{=} \{c^H B_H(t)\}_{t\ge 0}$.  One necessarily has that
$H\in (0,1)$ and, for some $\sigma^2 = {\rm Var}(B_H(1))>0$:
$$
 {\rm Cov}(B_H(t), B_H(s) ) = \frac{\sigma^2}{2}{\Big(} |t|^{2H} + |s|^{2H} - |t-s|^{2H} {\Big)},\ t,s \ge 0
$$
(see e.g.\ \cite{samorodnitsky:taqqu:1994book}).

Relation \eqref{e:fast-regime} shows that the fluctuations of the cumulative traffic about its mean behave
asymptotically like the fractional Brownian motion, as the number of users $M$ and the time scale $T$ are 
sufficiently large.  The increments $G(k):= B_H(k) - B_H(k-1),\ k= 1,2,\dots,$ of fBm then can then serve as
a model for the traffic traces of the number of bytes transmitted over the network over certain, sufficiently large
time scales. 

 The order of the limits in \eqref{e:fast-regime} is important.  If one takes $T\to\infty$ first and then 
$M\to\infty$, as shown in \cite{taqqu:willinger:sherman:1997}, one obtains:
\begin{equation}\label{e:slow-regime}
 {\cal L}\lim_{M\to\infty} \frac{1}{M^{1/\alpha}} {\Big\{} {\cal L}\lim_{T\to\infty} \frac{1}{T^{1/\alpha}}
  X_0^*(Tt),M {\Big\}}_{t\ge 0} = \{\Lambda_\alpha(t)\}_{t\ge 0}.
\end{equation}
Now the limit process $\Lambda_\alpha = \{\Lambda_\alpha(t)\}_{t\ge 0}$ has {\it independent} and stationary 
increments with $\alpha-$stable distributions, with $\alpha$ being as in \eqref{e:alphas}.
 It is the L\'evy stable motion -- the infinite variance counterpart 
to the Brownian motion. 

Relations \eqref{e:fast-regime} and \eqref{e:slow-regime} show two different regimes for the network.  The 
first involves many users relative to the time scale and the second, just a few users relative to the time scale.
As shown in \cite{mikosch:resnick:rootzen:stegeman:2002} (see also 
\cite{gaigalas:kaj:2003,pipiras:taqqu:levy:2004,mikosch:samorodnitsky:2007}), one can consider the limit when
the number of users $M=M(T)$ grows to infinity, as a function of the time scale $T$.  Then:

\begin{itemize}
 
 \item {\em (fast growth)} If ${(M(T)T)^{1/\alpha}}/{T} \to \infty,$ as $T\to\infty$,
then 
$$
{\cal L}\lim_{T\to\infty}
 {\Big\{} \frac{1}{T^{H} \sqrt{M(T)}} X_0^{*}(Tt,M)  {\Big\}}_{t\ge 0} = \{B_H(t)\}_{t\ge 0}.
$$

\item {\em (slow growth)} If ${(M(T)T)^{1/\alpha}}/{T} \to 0,$ as $T\to\infty$,
then 
$$
{\cal L}\lim_{T\to\infty}
 {\Big\{} \frac{1}{(TM(T))^{1/\alpha}}X_0^{*}(Tt,M)  {\Big\}}_{t\ge 0}
  = \{\Lambda_\alpha(t)\}_{t\ge 0}.
$$
\end{itemize}

The {\em fast growth} scenario shows that if the number of users $M(T)$ grows relatively fast, then the same 
limit as in \eqref{e:fast-regime} is achieved.  The {\em slow growth} regime on the other hand, yields the stable 
L\'evy motion in the limit, when there are relatively few users sharing the link.
The intermediate regime when ${(M(T)T)^{1/\alpha}}/{T} \to c\in(0,\infty)$ is considered in \cite{gaigalas:kaj:2003}.

This abundant theory offers a multitude of tools for modeling the temporal dependence of traffic traces
in various regimes.  For example, the users need not be of the same type.  As in 
\cite{dauria:samorodnitsky:2005} one may consider $q$ classes of users $M_k,\ 1\le k\le q$, where
$M = \sum_{k=1}^q M_k$, and $M_k(T)\to \infty$ as $T\to\infty.$  The users within a given class are of the 
same type with parameters $\alpha_k \in (1,2),$ and $H_k:= (3-\alpha_k)/2,\ 1\le k\le q$.  By balancing the
rates of the $M_k(T)$'s one can obtain in the limit
$$
{\cal L}\lim_{T\to\infty} \{\frac{1}{a(T)} X_0^*(Tt,M) \}_{t\ge 0} = \sum_{k\in {\cal F} } B_{H_k} 
 + \sum_{k \in {\cal S}} \Lambda_{\alpha_k},
$$
where $B_{H_k} = \{B_{H_k}(t)\}_{t\ge 0}$ and $\Lambda_{\alpha_k} = \{\Lambda_{\alpha_k}(t)\}_{t\ge 0},\
1\le k\le q$ are independent fBm's and L\'evy stable motions, respectively.  Here
$\{1,\cdots,q\} = {\cal F}\cup {\cal S}$ is the partition of the groups of users into
subsets of {\em fast} and {\em slow} growth regimes, respectively.

Similar results were shown to hold for the $M/G/\infty$ and other activity rate models (see e.g.\
\cite{mikosch:samorodnitsky:2007}).
%
%
\medskip
\noi{\em Remarks.}
\begin{enumerate}
 \item If the individual user behavior is modeled by $M/G/\infty$ processes with heavy--tailed, infinite variance
 durations $Z_j$'s, then similar asymptotic results hold for the cumulative traffic fluctuations.  In fact,
 as shown in \cite{mikosch:samorodnitsky:2007}, this is so for the general activity rate model in 
 \eqref{e:act-rate}.  

 \item As argued above, by balancing the rates of multiple groups of users, one can obtain complex hybrid models,
 composed of fBm's and L\'evy stable motions.  In practice, however, typically one component dominates.  In fact,
 as shown in \cite{mikosch:samorodnitsky:2007}, the fBm limit is more robust than the L\'evy stable motion with 
 respect to the type and the regimes of the activity rate models considered.

 In fact, the fundamental theorem of Lamperti (see eg Theorem 2.1.1 in \cite{embrechts:maejima:2002}) implies 
 an interesting {\em robustness} and {\em homogeneity} property.  Namely, suppose that $X^{(i)}(t)$'s are all
{\em stationary } in $t$.  If the the time--scale limit
$$
 {\cal L}\lim_{T\to\infty} \{ \frac{1}{a(T)} X_0^{*}(Tt,M) \}_{t\ge 0}  = \{\xi(t,M)\}_{t\ge 0},
$$
is non--trivial, then it is necessarily {\em self--similar}.  That is,
$\{\xi(ct,M)\}_{t\ge 0} \stackrel{d}{=} \{c^H \xi(t,M)\},$ for all $c>0$ with some $H>0$.

This implies that if the number of users $M$ is either fixed or already large enough for the Gaussian 
asymptotics to hold, then the time--scaling limit is necessarily either a single L\'evy stable motion,
or a single fractional Brownian motion.

Thus the complex hybrid models involving sums of multiple fBm's and L\'evy stable motions are rather fragile.
That is, they may occur only if a careful balance between the rate of growth of the users and the time--scale 
is imposed.  In reality, the single--fBm and single--L\'evy stable motion provide good, first--order limit 
approximations of traffic fluctuations that remain valid under changes of time--scales.

This observation is the reason why we advocate studying first the simpler, self--similar models involving either
a single fBm or a single L\'evy stable motion.  Accounting for the hybrid models involves careful considerations 
of time--scales, which presents formidable statistical challenges.

\end{enumerate}

\section{Network--Wide Traffic Modeling}
 \label{s:network-wide}

   \subsection{Asymptotic Approximations}
 
 As discussed in the introduction, we assume that traffic is {\em fluid} and it propagates instantaneously
 through the network so that the {\em routing equation} \eqref{e:Y=AX} holds.  As in Section \ref{s:single-link},
 we model the traffic intensity $X_j(t)$ over route $j$ as a composition of $M_j$ independent users.  We suppose, 
 in addition that the $X_j(t)$'s are independent in $j$ and composed of $M_j$ independent and identically distributed 
 (i.i.d.) {\em On/Off} sources: 
 \begin{equation}\label{e:Xj=On/Off}
   X_j(t) = \sum_{i=1}^{M_j} X_{j}^{(i)} (t),\ \ 1\le j\le {\cal J}.
\end{equation}
 We then obtain the following results:

\medskip
 \begin{theorem} \label{t:fast-regime}
 Let $X_j(t)$'s be as in \eqref{e:Xj=On/Off}, where 
 the {\em On/Off components} have common parameter $\alpha$ as in \eqref{e:alphas}. 
 Suppose that $M_j \sim r(j) M,\ M\to\infty$, for some constants $r(j)>0$ and let $\vec Y(t)$ be as in \eqref{e:Y=AX}. 
 Then,
\begin{eqnarray}\label{e:fast-Y}
 & & {\cal L}\lim_{T\to\infty} {\cal L}\lim_{M\to\infty} 
   \frac{1}{T^H \sqrt{M}}\int_{0}^{Tt} (\vec Y(\tau) - \E\vec Y(0))d\tau \nonumber\\
& & \quad\quad\quad\quad\quad\quad\quad\quad   = \{A \vec B_H(t)\}_{t\ge 0},
\end{eqnarray}
where $\vec B_H(t) = (r(j) B_H^{(j)}(t) )_{1\le j\le {\cal J}}$ and $B_H^{(j)}(t)$'s are i.i.d.\ fBm's 
with parameter $H = (3-\alpha)/2 \in (1/2,1)$.
\end{theorem}

 \begin{theorem} \label{t:slow-regime} Under the conditions of Theorem \ref{t:fast-regime}, we have
\begin{eqnarray}\label{e:slow-Y}
 & & {\cal L}\lim_{M\to\infty} {\cal L}\lim_{T\to\infty} 
   \frac{1}{T^{1/\alpha} M^{1/\alpha} }\int_{0}^{Tt} (\vec Y(\tau) - \E\vec Y(0))d\tau\nonumber \\
& & \quad\quad\quad\quad \quad\quad\quad\quad  = \{A \vec \Lambda_\alpha(t)\}_{t\ge 0},
\end{eqnarray}
where $\vec \Lambda_\alpha(t) = (r(j) \Lambda_\alpha^{(j)} (t) )_{1\le j\le {\cal J}}$ 
and $\Lambda_\alpha^{(j)}(t)$'s are i.i.d.\ L\'evy $\alpha-$stable motions.
\end{theorem}

Theorems \ref{t:fast-regime} and \ref{t:slow-regime} correspond, respectively,
to the {\em fast} and {\em slow} regime asymptotics in the single--flow case.  Their proofs follow readily 
from the well--known single--flow results with an application of the continuous mapping theorem.

   \subsection{A representation via functional L\'evy and 
      functional fractional Brownian motions} \label{s:f-fBm,f-Lsm}

 In this section, we introduce two classes of stochastic processes, indexed by functions, which can be used to
 succinctly represent the limit processes arising in Theorems \ref{t:fast-regime} and \ref{t:slow-regime}.
 The purpose of this more abstract treatment is to develop tools and insight that can be used in statistical 
 inference for the network models.

\noi {\bf Functional fBm:} Consider a measure space $(E,\mu)$ and the set of functions 
$$L^{2H}(\mu) = \{f:E\to\bbR,\ \|f\|_{2H}^{2H} := \int_E |f|^{2H} d\mu <\infty\},
$$
where $H\in (0,1)$.  Introduce the functional
\begin{equation}\label{e:phi}
 \phi_{2H}(f,g) := \frac{\sigma^2}{2} {\Big(} \|f\|_{2H}^{2H} + \|g\|_{2H}^{2H} - \|f-g\|_{2H}^{2H} {\Big)},
\end{equation}
for $f,g\in L^{2H}(\mu)$ and $\sigma>0$.

The functional $(f,g)\mapsto \phi(f,g)$ resembles the auto--covariance function of an fBm.  It turns out that
$\phi(f,g)$ is positive semi--definite (see Proposition \ref{p:pos-def} in the Appendix).  One can thus 
define a Gaussian process with covariance $\phi_{2H}$:

\smallskip
 \begin{definition} Let $0< H\le 1$.
 A zero mean Gaussian process $B = \{B(f)\}_{f\in L^{2H}(\mu)}$ indexed by the functions
 $f\in L^{2H}(\mu)$ is said to be a {\em functional fractional Brownian motion} (f--fBm), if:
 $$
  {\rm Cov}(B(f),B(g)) = \E B(f) B(g) = \phi_{2H}(f,g),\ \ f,g\in L^{2H}(\mu).
 $$
\end{definition}

It turns out that the limit process in Theorem \ref{t:fast-regime} can be expressed in terms of a 
{\em functional fBm}. Indeed, let $E = \{1,\cdots, {\cal J}\}$ and let the measure $\mu$ be the counting
measure on $E$.   Consider the f--fBm $B = \{B(f)\}_{f\in L^{2H}(\mu)}.$  

\smallskip
\begin{proposition}\label{p:f-fBm} For the limit process in \eqref{e:fast-Y}, we have
$$
 \{ A \vec B_H(t) \}_{t\ge 0}  \stackrel{d}{=}  {\Big\{}  (B(t f_\ell))_{1\le \ell\le L} {\Big\}}_{t\ge 0}.
$$ 
Here $f_\ell (u) = r(u)^{1/H} 1_{A_\ell}(u),$ where $A_\ell\subset \{1,\cdots,{\cal J}\}$ denotes the set 
of routes that use link $\ell,\ 1\le \ell \le L$ and $r(u),\ 1\le u \le {\cal J}$ is as in 
Theorem \ref{t:fast-regime}.
\end{proposition}

\noi 
The proof is given in the Appendix.  The next result shows the basic properties of f--fBm's.

\begin{proposition}\label{p:properties} Let $H\in (0,1]$ and 
$B = \{B(f)\}_{f\in L^{2H}(\mu)}$ be f--fBm.

{\it (i)} The process $B$ is $H-$self--similar:
 \begin{equation}\label{e:H-ss}
  \{B(cf)\}_{f\in L^{2H}(\mu)} \stackrel{d}{=}\{c^H B(f)\}_{f\in L^{2H}(\mu)},\ \ (\forall c>0),
 \end{equation}
where $\stackrel{d}{=}$ denotes equality of the finite--dimensional distributions.

{\it (ii)} The process $B$ has stationary increments:
\begin{equation}\label{e:stat-inc}
 \{B(f+h) - B(h)\}_{f\in L^{2H}(\mu)} \stackrel{d}{=} \{ B(f)\}_{f\in L^{2H}(\mu)},
\end{equation}
for all $h\in L^{2H}(\mu)$.

{\it (iii)} If $fg =0$ $\mu-$a.e., then $B(f)$ and $B(g)$ are independent.

{\it (iv)}  $B_f(t):= B(tf),\ t\in\bbR$ is an ordinary fBm process.

{\it (v)} If $ H \not=1$, then $B(f) + B(g) = B(f+g)$, almost surely, {\em if and only if}
 $fg=0$, $\mu-$a.e.  (Note that by {\it (ii)} above, we always have
 $B(f) - B(g) \stackrel{d}{=} B(f-g)$.)
\end{proposition}

\noi 
The proof is given in the Appendix.

Now, to gain more intuition behind the role of f--fBm in representing the limit process in Theorem 
\ref{t:fast-regime}, suppose that $r(j) = 1$ therein, i.e.\ all routes involve the same number 
of users $M_j = M$.  Consider the random variables $B(t f_{\ell_1} )$ and $B(s f_{\ell_2} )$ representing the
asymptotic cumulative fluctuations of traffic over links $\ell_1$ and $\ell_2$ respectively.  
Since $f_\ell = 1_{A_\ell}$ is merely an indicator function, we have:
\begin{eqnarray}\label{e:space-time-cov}
& & \E B(tf_{\ell_1}) B(s f_{\ell_2}) = 
 \frac{\sigma^2}{2}{\Big(} |t|^{2H}\mu(A_{\ell_1} ) + |s|^{2H} \mu(A_{\ell_2})\nonumber\\
& & \quad \quad \quad -|t-s|^{2H}\mu(A_{\ell_1}\cap A_{\ell_2}) - |t|^{2H} \mu(A_{\ell_1}\setminus A_{\ell_2})\nonumber\\
& & \quad\quad\quad - |s|^{2H} \mu(A_{\ell_2}\setminus A_{\ell_1}) {\Big)}\nonumber\\
&  & = \mu(A_{\ell_1}\cap A_{\ell_2}) \frac{\sigma^2}{2}(|t|^{2H} + |s|^{2H} - |t-s|^{2H}).
\end{eqnarray}
Recall that $A_{\ell}\subset\{1,\cdots,{\cal J}\}$ is the set of all routes that involve link $\ell$.  Thus,
the last relation has the following natural interpretation.  The spatial dependence between the links $\ell_1$ and 
$\ell_2$ is governed solely by the routes they have in common, i.e.\ the set $A_{\ell_1}\cap A_{\ell_2}$.
On the other hand, the temporal dependence follows the fBm model.  In particular,
 $B(t f_{\ell_1})$ and $B(t f_{\ell_2})$ are independent if and only if links 
$\ell_1$ and $\ell_2$ have no common routes, i.e.\ $\mu(A_{\ell_1}\cap A_{\ell_2}) =0$.

\medskip
\noi{\bf Functional L\'evy stable motion:} As for f--fBm, consider the measure space $(E,\mu)$ and
 the set of functions $L^{1}(\mu)$.

 \begin{definition} Let $\alpha\in (1,2)$. Consider a zero--mean $\alpha-$stable measure
 $M_\alpha(dx,du)$ on $\bbR \times E$ with control measure $dx \times \mu(du)$ (see \cite{samorodnitsky:taqqu:1994book}).
 Let
$$
 \Lambda(f) := \int_{\bbR\times E} {\Big(}1_{(-\infty,f(u)]}(x) - 1_{(-\infty,0)}(x) {\Big)} M_\alpha(dx,du),
$$
for any $f\in L^1(\mu)$. The process $\{\Lambda(f)\}_{f\in L^1(\mu)}$, 
indexed by the functions $f\in L^\alpha(\mu)$ is said to be a {\em functional L\'evy stable motion} (f--Lsm).
\end{definition}

As for f--fBm, we have:

\begin{proposition}\label{p:f-Lsm} For the limit process in \eqref{e:slow-Y}, we have
$$
 \{ A \vec \Lambda_\alpha(t) \}_{t\ge 0}  \stackrel{d}{=}
   {\Big\{}  (\Lambda(t f_\ell))_{1\le \ell\le L} {\Big\}}_{t\ge 0}.
$$ 
Here $f_\ell (u) = r(u)^{1/\alpha} 1_{A_\ell}(u),$ where $A_\ell\subset \{1,\cdots,{\cal J}\}$ is the set 
of routes involving link $\ell,\ 1\le \ell \le L$ and $r(u),\ 1\le u \le {\cal J}$ is as in 
Theorem \ref{t:fast-regime}.
\end{proposition}

The properties of the f--Lsm parallel those of f--fBm. For example, the process $\{\Lambda(tf)\}_{t \ge 0}$
is a L\'evy stable motion.

\begin{proposition}\label{p:properties-fLsm} Let $\alpha \in (1,2)$ and $\{\Lambda(f)\}_{f\in L^1(\mu)}$ be a functional L\'evy $\alpha$-stable
motion.  We then have:

{\it (i)} The process $\Lambda$ is $1/\alpha$ self--similar:
$$
\{\Lambda(cf)\}_{f\in L^1(\mu)} \stackrel{d}{=}\{c^{1/\alpha} \Lambda (f)\}_{f\in L^1(\mu)},\ \ (\forall c>0).
$$

{\it (ii)} $\Lambda$ has stationary increments:
$$
\{\Lambda(f+h)-\Lambda(h)\}_{f\in L^1(\mu)} \stackrel{d}{=}\{ \Lambda(f)\}_{f\in L^1(\mu)},\ \ (\forall h\in L^{1}(\mu)).
$$ 

{\it (iii)} $\Lambda(f)$ and $\Lambda(g)$ are independent {\em if and only if} $fg \le 0$ $\mu-$a.e.

{\it (iv)} For all $0\le f_1 \le f_2 \le \cdots \le f_n$ (mod $\mu$) and $n\in\bbN$, the increments
$$
 \Lambda(f_1), \Lambda(f_2) -\Lambda(f_1),\cdots, \Lambda(f_n) -\Lambda(f_{n-1}),
$$
are {\em independent}.

{\it (v)} $\Lambda(f+g) = \Lambda(f) + \Lambda(g)$, {\em if and only if} $fg = 0$ $\mu-$a.e.

{\it (vi)} $\{\Lambda(tf)\}_{t\ge 0}$ is an ordinary L\'evy $\alpha$--stable motion.
\end{proposition}

\medskip
\noi The proof is given in the Appendix.

\medskip
\noi{\em Remark:} Here, for simplicity, we focus only on the case $\alpha\in (1,2)$, where the mean of $M_\alpha$ is finite and set to zero.  Implicitly,
the skewness coefficient function is assumed to be constant.  The functional L\'evy stable motion can be defined for all $\alpha\in (0,2)$, provided that the
random measure $M_\alpha$ is {\em strictly stable} with constant skewness intensity function.  For example, the symmetric $\alpha-$stable case is particularly
simple.  For more details of $\alpha-$stable stochastic integration, see eg \cite{samorodnitsky:taqqu:1994book}.

\medskip
\noi{\bf Integral representation of f--fBm:} The explicit representation of the f--Lsm processes suggests
that the f--fBm may be also conveniently handled through stochastic integrals.  Indeed:

\medskip
\begin{proposition}\label{p:f-fBm-int-rep} For all $H\in (0,1),$ we have that:
$$
 B(f) := \int_{\bbR\times E} {\Big(}(f(u)-x)_+^{H-1/2} - (-x)_+^{H-1/2}{\Big)} W(dx,du),
$$
is a functional fBm, where $W(dx,du)$ is a Gaussian random measure with control measure $dx\times \mu(du)$ and
$(x)_+:=\max\{x,0\}$.
\end{proposition}

\medskip
\noi The proof is given in the Appendix.

The last representation provides further tools as well as intuition into the nature of the f--fBm.  Indeed,
suppose that $E$ is discrete.  Then, $W(dx, \{u_1\})$ and $W(dx, \{u_2\})$ are independent Gaussian measures 
on $\bbR$, for $u_1 \not = u_2$. Thus, the stochastic integral over $E$ becomes a sum of independent processes,
each of which has the form of a fractional Brownian motion.  That is,
$$
 B(f) = \sum_{u\in E} B_H^{(u)} (f(u)),
$$
where $\{B_H^{(u)}(t)\}_{t \in \bbR}$ are i.i.d.\ fBm's indexed by $u\in E$.

Thus, the functional fBm may be viewed as a suitable, infinitesimal sum of independent fBm's each 
indexed by the corresponding values $f(u)$ of the functional argument $f$.  This is essentially why
the f--fBm provides a succinct representation of the limit process in Theorem \ref{t:fast-regime}.

In the next section, we utilize the simple parametric form of the limit approximations to solve the
{\em network kriging} problem.
 
\section{An Application to Network Kriging} \label{s:kriging}

In view of Theorems \ref{t:fast-regime} and \ref{t:slow-regime} one can model the joint distribution
of the traffic traces $Y_\ell(t),\ 1\le \ell \le L$ as increments of functional fBm or functional 
L\'evy stable motion.  Here, we focus on the {\em fast regime}, where according to Theorem 
\ref{t:fast-regime}, the traffic traces are approximated by Gaussian processes.  

Consider the traffic time series 
$$
 Y_{\delta}(\ell,k):=\int_{(k-1)\delta}^{k\delta} Y_\ell(t) dt,\ k=1,2,\cdots
$$
of the number of bytes traversing link $\ell$ during the $k$--th time interval $((k-1)\delta,k\delta)$, 
for a fixed time scale $\delta>0$. Guided by the multiple sources asymptotics, let $B(f),\ f\in L^{2H}(E,\mu)$ be
an f--fBm, where $E=\{1,\cdots,{\cal J}\}$ and $\mu$ is the counting measure on $E$.  Set,
$$
  Y_\delta(\ell, k) := \mu_Y(\ell) + B(k f_\ell) - B((k-1)f_\ell),\ \ \ k =1,2,\cdots,
$$ 
where $f_\ell(u) = r(u)^{1/H} 1_{A_{\ell}}(u),\ u\in E\equiv \{1,\cdots,{\cal J}\}$ and $A_\ell$ is the set of 
all routes using link $\ell$.  Here $\mu_Y(\ell) = \E Y_\delta (\ell,k)$ is the traffic mean over link $\ell$.

Assuming that the mean structure $\vec \mu_Y = (\mu_Y(\ell))_{1\le \ell \le L}$ and the parameters $H$ and
$r(u)$ of the limit f--fBm model are {\em known}, then one recovers the joint distribution of the traffic 
load on the network across {\em all links} $\ell$ and {\em time slots} $k$.  This allows one to address a number of
fundamental statistical problems.  

\medskip
\noi{\bf Instantaneous prediction (network kriging):}  {\em Observed are the traffic loads 
\begin{equation}\label{e:D}
 {\cal D}:= \{ Y(\ell,t),\ 1\le t\le t_0,\ \ell \in {\cal O} \},
\end{equation}
over the set of links $\ell \in {\cal O} \subset \{1,\cdots,L\}$ at time slots $t,\ 1\le t\le t_0$.
Predict the traffic load $\what Y(\ell_0,t_0)$ on a {\em unobserved link} $\ell_0$, in 
terms of the data ${\cal D}$.}
 
\smallskip
\noi{\bf Spatio--temporal prediction:}  {\em Given the data ${\cal D}$ in \eqref{e:D}, predict the traffic load 
$\what Y(\ell,t_0+h)$ on a {\em observed} or {\em unobserved link} $\ell$, at some future time $t_0+h>t$.
}

\smallskip
\noi{\em Remarks:} \begin{enumerate}

\item The estimation of the Hurst parameter $H$ is a well--studied problem (see e.g.\ 
\cite{taqqu:teverovsky:1996S,abry:veitch:1998,bardet:lang:oppenheim:philippe:stoev:taqqu:2003-livre,
 stoev:taqqu:park:marron:2005,stoev:taqqu:park:michailidis:marron:2006}.) 
We advocate the use of robustified {\em wavelet} methods to obtain $H$ in practice (see 
e.g. \cite{abry:veitch:1998,stoev:pipiras:taqqu:2002,stoev:taqqu:2003w,stoev:taqqu:2005A,stoev:taqqu:park:marron:2005}.)

On the other hand, the estimation of the mean structure $\vec \mu_Y = (\mu_Y(\ell))_{1\le \ell \le L}$, and the 
underlying parameter $r(u), 1\le u \le {\cal J}$ in the covariance structure are important and challenging problems
in practice.  We address these problems in a general statistical framework with the help of latent models and
auxiliary NetFlow data sets in the forthcoming work \cite{vaughan:stoev:michailidis:2009}.

 \item In the interest of space, we focus only on the first, instantaneous prediction problem.  The $h$--step 
prediction problem can be addressed similarly (see e.g.\ \cite{vaughan:stoev:michailidis:2009}.)  

We refer to the instantaneous prediction as {\em network kriging} because of its resemblance to geostatistical
prediction problems.  The term network kriging was introduced first, to best of our knowledge, by Chua, 
Kolaczyk and Crovella in \cite{chua:kolaczyk:crovella:2006} in the context of predicting eg delays along routes
from active network measurements of flows in the network.  Here, our setting is different since the focus is
{\em link} rather than {\em flow} measurements.
 
\end{enumerate}

For simplicity, let $\delta =1$, time $t\in \bbN$ be discrete, and (with some abuse of notation)
$$
 \vec Y(t):= (Y_\delta(\ell,t))_{1\le \ell \le L},\ \ t = 1,2,\cdots. 
$$
Partition the vector $\vec Y(t)$ and the rows of the routing matrix $A$ into two components, 
corresponding to the indices of the {\it unobserved} ('u') and {\it observed} ('o') sets of links:
$$
\vec Y(t) = \left( \begin{array}{l}
                   Y_u(t)\\ Y_o(t)\end{array}\right)\ \ \  \mbox{ and }\ \ \
A =  \left( \begin{array}{l}
                   A_u \\ A_o\end{array}\right).
$$

\begin{proposition} \label{p:std-Kriging} Let $\vec Y(t) = A \vec X(t)$, where $\E \vec X(0) = \mu_X$,
and $\Sigma_{X} := \E (\vec X(0) - \mu_X) (\vec X(0) - \mu_X)^t$.  Suppose that the matrix 
$A_o\Sigma_XA_o^t$ is invertible.  Then:

{\it (i)} The statistic
\begin{equation}\label{e:p:std-Kriging-Y-hat}
 \what Y_u(t_0) = A_u\mu_X+ A_u\Sigma_XA_o^t (A_o\Sigma_XA_o^t)^{-1}(Y_o(t_0)-A_o\mu_X)
\end{equation}
is a unbiased predictor for $Y_u(t_0)$ in terms of the data ${\cal D}$ in \eqref{e:D}. The 
mean--squared error (m.s.e.) matrix of $\what Y_u(t_0)$ is:
\begin{eqnarray} \label{e:p:std-Kriging-MSE}
& & \MSE(\what Y_u(t_0) | {\cal D}) \\ 
& & \ \ := \E {\Big(} (\what Y_u(t_0)-Y_u(t_0))(\what Y_u(t_0)-Y_u(t_0))^t | {\cal D} {\Big)} \nonumber\\
  & & \ \ = A_u\Sigma_XA_u^t-A_u\Sigma_XA_o^t(A_o\Sigma_XA_o^t)^{-1}A_o\Sigma_XA_u^t, \nonumber
\end{eqnarray}
where the last expectation is conditional, given the data ${\cal D}$.

{\it (ii)} The statistic $\what Y_u(t_0)$ in \eqref{e:p:std-Kriging-Y-hat} is the unique best unbiased
m.s.e.\ predictor of $Y_u(t_0)$ in terms of the data ${\cal D}$ in \eqref{e:D}.  That is,
for any other unbiased predictor $Y_u^*(t_0)$, we have that
\begin{equation}\label{e:p:std-Kriging-MSE-1}
 \MSE(\what Y_u(t_0) | {\cal D}) \le \MSE(Y_u^*(t_0) | {\cal D})
\end{equation}
where the last inequality means that the difference between the matrices in the right-- and
the left--hand sides is positive semidefinite.
\end{proposition}

\noi The proof is given in the Appendix. We now make a few important observations.

\medskip
\noi{\em Remarks:}
\begin{enumerate}

 \item If $\vec Y(t)$ is non--Gaussian, then the estimator in \eqref{e:p:std-Kriging-Y-hat} remains
the {\em best linear unbiased predictor} (b.l.u.p.) of $Y_u(t_0)$ in terms of the data ${\cal D}$.
Relations \eqref{e:p:std-Kriging-MSE} and \eqref{e:p:std-Kriging-MSE-1} continue to hold,
where now $Y_u^*(t_0)$ is an arbitrary linear in ${\cal D}$, unbiased predictor of
$Y_u(t_0)$.

 \item By Gaussianity, it is easy to see that $\what Y_u(t_0)$ in \eqref{e:p:std-Kriging-Y-hat} also 
 maximizes the conditional likelihood of $Y_u(t_0)$, given the data.

 \item Note that {\em only} the observations $Y_o(t_0)$ at the {\em present time} $t_0$ are involved in
\eqref{e:p:std-Kriging-Y-hat}.  This is due to the product form of the space--time covariance
structure of the functional fBm \eqref{e:space-time-cov} and Proposition \ref{p:space-time-factor} below.

 \item If the matrix $A_o\Sigma_XA_o^t$ is singular, then one can replace the inverse in
\eqref{e:p:std-Kriging-Y-hat} and \eqref{e:p:std-Kriging-MSE} by the Moore--Penrose generalized inverse.
This corresponds to focusing on the range of $A_0\Sigma_X A_o^t$, where the latter matrix is invertible.
The statistic $\what Y_u(t_0)$ remains the b.l.u.p. In practice, $A_0\Sigma_X A_o^t$ is singular only
when the traffic over a link is a perfect linear combination of the traffic over another set of links. 
This occurs in tree--type topologies, for example, where the internal nodes do not generate traffic.

 \item In the {\em slow regime} (Theorem \ref{t:slow-regime}) the functional Lsm infinite variance model
 for $\vec Y(t)$ should be used.  The prediction problems can then be also addressed but not with respect to
 the square loss.  One can consider minimizing $\E |\what Y_{\rm u}(t_0) - Y_{\rm u}(t_0)|^{p}$ for 
 $p<\alpha$ or, equivalently, the scale coefficient of the $\alpha$--stable variable. In this case, no closed--form 
 solutions are available but one can obtain numerical expressions for the best linear predictors.  
 Our experiments indicate that the coefficients of these linear predictors are often very close to those of 
 the least squares predictor in \eqref{e:p:std-Kriging-Y-hat}.
\end{enumerate}

\medskip
The fact that the b.l.u.p.\ $\what Y_u(t_0)$ in Proposition \ref{p:std-Kriging} does not depend
on the past data $Y_o(t),\ t< t_0$ shows that the $\what Y_u(t_0)$ is in fact the {\it standard kriging}
predictor, which is well--studied in spatial statistics (see eg \cite{cressie:1993}). We shall therefore 
refer to $\what Y_{u}(t_0)$ as to the {\it standard network kriging predictor}.

\medskip
The following result provides the general solution to $h$--step prediction problem.  We start by introducing 
some notation. Consider the Toeplitz matrix:
$$
 \Gamma_{m+1} := (\gamma_X(|i-j|))_{0\le i,j\le m}
$$
and  the vector $\vec \gamma_{m+1}(h) = (\gamma_X(h+j))_{0\le j \le m}$,
where
$$
\gamma_X(k) = \frac{\sigma^2}{2} {\Big(} |k+1|^{2H} +|k-1|^{2H} - 2|k|^{2H} {\Big)}.
$$
Since $\gamma_X(0) = \sigma^2>0$ and $\gamma_X(k) \to 0,$ as $k\to\infty$, the matrix 
$\Gamma_{m+1}$ is invertible, for all $m\in\bbN$ (see eg Proposition 5.1.1 in
\cite{brockwell:davis:1991}).

\begin{proposition}\label{p:temporal-pred}
 Assume the conditions of Proposition \ref{p:std-Kriging}. Let $\mu_i = A_i \mu_X 
=\E Y_i(t),\ i \in\{\mbox{'u','o'}\}$.

{\it (i)} The statistic 
\begin{equation}\label{e:temporal-pred-i}
 \what Y_o(t_0+h) := \mu_o + \sum_{j=0}^m c_j(h) (Y_o(t_0-j) - \mu_o), 
\end{equation}
is a unbiased predictor of $Y_o(t_0+h),\ h\ge 1$ via ${\cal D}$, where $\vec c(h) \equiv (c_j(h))_{0\le j\le m} 
= \Gamma_{m+1}^{-1} \vec \gamma_{m+1} (h)$.  The m.s.e.\ matrix of $\what Y_o(t_0+h)$ is then:
\begin{equation}\label{e:temporal-pred-i-mse}
 \MSE(\what Y_o(t_0+h)|{\cal D}) = \sigma^2(h) A_o \Sigma_X A_o^t,
\end{equation}
where
$
 \sigma^2(h) := \gamma_X(0) - \vec c(h)^t \Gamma_{m+1} \vec c(h) \equiv 
1 - \vec\gamma_{m+1}(h)^t \Gamma_{m+1}^{-1} \vec\gamma_{m+1}(h). 
$

{\it (ii)} The statistic
\begin{equation}\label{e:temporal-pred-ii}
 \what Y_u(t_0+h) := \mu_u + C (\what Y_o(t_0+h) - \mu_o)
\end{equation}
is a unbiased predictor of $Y_u(t_0+h)$ via ${\cal D}$, 
where $C = A_u\Sigma_XA_o^t(A_o\Sigma_XA_o^t)^{-1}$ and $\what Y_o(t_0+h)$ is as
in \eqref{e:temporal-pred-i}.  The m.s.e.\ matrix of $\what Y_u(t_0+h)$ is:
$$
\MSE(\what Y_u(t_0+h)|{\cal D}) = \sigma^2(h) C A_o \Sigma_X A_o^t C^t + 
   \MSE(\what Y_u(t_0)|{\cal D}),
$$
where $C$ is as in \eqref{e:temporal-pred-ii} and where $\MSE(\what Y_u(t_0)|{\cal D})$
is as in \eqref{e:p:std-Kriging-MSE}.

{\it (iii)} The statistics in {\it (i)} and {\it (ii)} yield the best m.s.e.\ predictors in the sense of 
Proposition \ref{p:std-Kriging} {\it (ii)}.  If $Y(t)$ is non--Gaussian, then these predictors are b.l.u.p.\ 
in terms of the data ${\cal D}$.
\end{proposition}

\noi The proof is given in the Appendix.

\medskip
The above results provide, in principle, complete solutions to the kriging and the $h$--step prediction
problems outlined above.  The underlying mean $\mu_Y = A \mu_X$, spatial $\Sigma_Y = A \Sigma_X A^t$
and temporal covariance structure, however, involves unknown parameters.  Moreover, their estimation from
link measurements is impossible, without network--specific regularity conditions, since the number of links is
typically much smaller than the number of routes ($L<<{\cal J}$).  
In \cite{vaughan:stoev:michailidis:2009}, we focus on designing suitable latent models
for the unknown means and covariances with the help of auxiliary NetFlow data on the route--level traffic.  
These models will involve a few parameters that can be successfully estimated from link measurements.

\section{Analysis of Internet2 Data}
 \label{s:NetFlow}

Here, we will first demonstrate the validity of our probabilistic models by using real network data.
We will then illustrate the performance of the {\em standard kriging predictor} in practice.
 
\begin{figure}
\begin{center}
\includegraphics[width=3in]{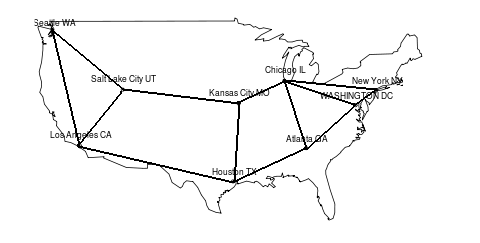}
{\caption{
{
\small The Internet2 backbone network consists of $9$ nodes and $26$ one--directional links.
All links have capacity of 10 Gbs/s, with the exception of the links: Chicago--Kansas, 
Kansas--Salt Lake City, New York--Washington, and Washington--Atlanta, all of which have
doubled capacity of 20 Gbs/s in each direction.}} \vskip-0.2in
\label{fig:I2-topology}}
\end{center}
\end{figure}

\medskip
\noi{\bf NetFlow data:} We obtained from \cite{I2}, sampled measurements of all packets traversing the Internet2 (I2)
backbone network (see Fig.\ \ref{fig:I2-topology}).  These data were used to {\em reconstruct} sampled versions 
of {\em all flows} $X_j(t),\ 1\le j \le {\cal J} = 72$ in I2.  Packet and bytes traces over the $100$
millisecond time scale were then obtained.  The routing matrix $A$ for the I2 network was deduced 
from these NetFlow data sets as well and it was found to be constant in time.

\begin{figure}
\begin{center}
\includegraphics[width=3in]{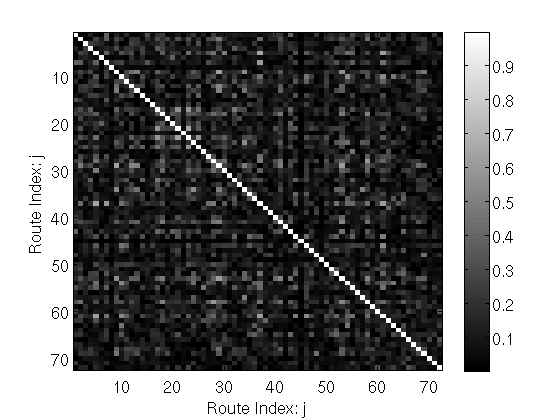}
{\caption{
{
\small Correlation matrix (absolute values) of the flow--level traffic computed from a typical hour--long traces
(in bytes per 10 sec).  The traces are deduced from NetFlow measurements on the Internet2 backbone on
Feb 19, 2009.  Brighter shades indicate numbers close to $1$. 
}} \vskip-0.2in
\label{fig:X-corr}}
\end{center}
\end{figure}

Computationally intensive processing is required to obtain the {\em flow level} data in practice.  Therefore, 
these data cannot be used directly for fast {\em on--line} prediction of traffic.  Nevertheless, we utilize this
information to validate the main assumptions in our models. Fig.\ \ref{fig:X-corr} indicates for example that
the $X_j(t)$'s are nearly uncorrelated in $j$, which supports the simplifying independence assumption. 
On the other hand, the wavelet spectrum of a typical flow indicates that $X_j(t)$ is well--modeled by a
fractional Gaussian noise time series for a wide range of time scales (see Fig.\ \ref{fig:Hurst}.)  Further, barring a few
anomalies in the NetFlow data, the Hurst exponents along most routes were found to be approximately equal (within
statistical significance).  These observations (along with NS2 simulation experiments, not discussed here due to lack space)
support the overall validity {\em global} functional fBm model for the cumulative traffic fluctuations.

%
%
\begin{figure}[h!]
\begin{center}
\includegraphics[width=3in]{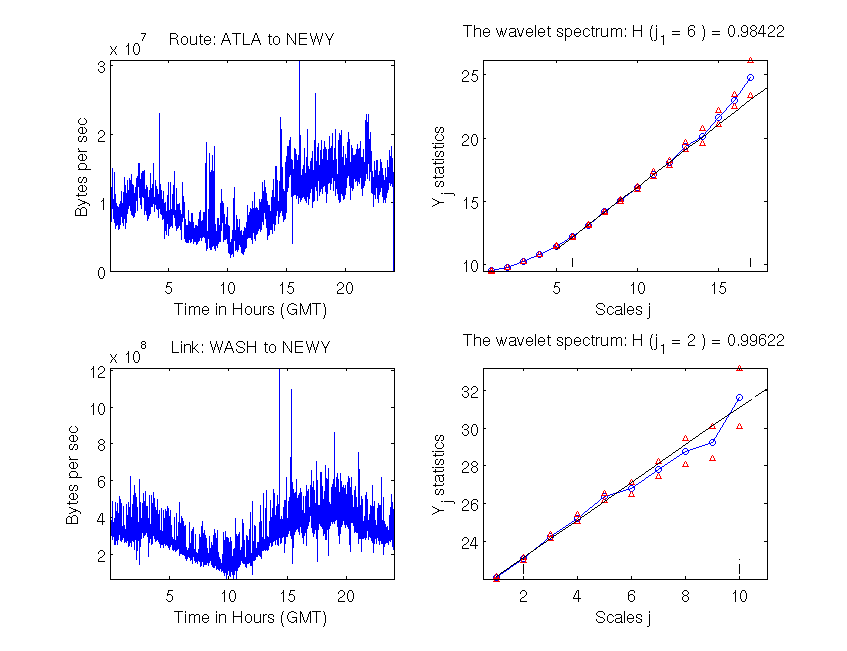}
{\caption{ \label{fig:Hurst}
{
\small {\it Top--left \& right}: a flow--level traffic trace (bytes per 100 msec)
from Atlanta to New York for March 17, 2009 {\em and} its wavelet spectrum, which yields a Hurst
parameter $\what H \approx 0.98$.
{\it Bottom--left \& right}: a link--level trace (bytes per 10 sec) for the Washington to New York link for March 17, 2009
{\em and} its wavelet spectrum with $\what H \approx 0.99$.  Observe the similarity between the diurnal patterns and
the Hurst exponents for the flow-- and link--level data. The linearity of the wavelet spectra confirms the 
fractional Gaussian noise model.}}} 
 \vskip -0.2in
\end{center}
\end{figure}

\medskip
\noi{\bf Traffic traces:} As indicated above, the NetFlow measurements cannot be used directly to readily 
predict the link loads {\em in real time}.  We acquired from \cite{I2} time--synchronized traffic traces of
packets and bytes on the $10$--second time scale, for all links in the Internet2 backbone network.  
As expected, since RTT $<< 10$ sec, the routing equation \eqref{e:Y=AX} can be safely assumed to hold
for the time scales of interest.  By using coarse--scale information obtained from the corresponding NetFlow data,
we approximated the mean $\mu_X$ and variance structure of the $X_j(t)$'s. Thus, by using $\mu_X$ and 
$\Sigma_X:= {\rm diag} (\sigma_{X_j}^2,\ 1\le j\le {\cal J}),$  we obtained from \eqref{e:p:std-Kriging-Y-hat} 
the {\em standard kriging estimator} for a number of scenarios with observed and unobserved links. 

Figs \ref{fig:pred_GOOD} and \ref{fig:pred_GOOD_zoom} demonstrate the success of our global modeling strategy in the
context of network kriging. By monitoring just a few links, with the help of the {\em standard kriging estimator}
described above, one can track relatively well the traffic load on other links.    Table \ref{tab:rmse} 
shows further that a given link can be relatively well predicted from measurements of as few as two other links.
The results also show that the choice of which set of link to monitor is an important design problem.

 \begin{table}[ht]
 \begin{center}
 \begin{tabular}{rrr}
 \hline
Number of links & Link labels & Relative m.s.e. \\
 \hline
2& 3,7 & 0.07 \\
2& 7,9 &  0.12 \\
2& 9,12 & 0.08 \\
2& 12,17 & 0.41 \\
2& 17,21 & 3.06 \\
2& 3,21 & 0.05 \\
3& 3,7,9 & 0.12 \\
4& 3,7,9,12  & 0.08 \\
5& 3,7,9,12,17  & 0.07 \\
6& 3,7,9,12,17,21 & 0.06 \\
8& 3,5,7,9,11,12,17,21  & 0.06 \\
10& 3,5,7,9,11,12,17,21,23,25 & 0.06 \\
  \hline
 \end{tabular}
 \end{center}
\caption{\label{tab:rmse}  Empirical relative mean squared error for the standard
kriging estimator:  ${\rm (r.m.s.e.)}\  = (\sum_{t=1}^{T} (\what Y_u(t) - Y_u(t))^2)/( \sum_{t=1}^{T} Y_u(t)^2)$.
The Internet2 backbone link 13 (Kansas City to Chicago) was predicted from various sets of other backbone
links. The data spans the entire day of February 19, 2009.  The link ID's are described in Table \ref{tab:I2links}.
}
 \end{table}

\begin{table}
\begin{tabular}{|lll||lll|}
\hline
\textbf{ID} & \textbf{Source--Destination} & \textbf{Cap.} &
\textbf{ID} & \textbf{Source--Destination} & \textbf{Cap.} \\
\hline
1 & Los Angeles--Seattle & 10 Gb/s & 2& Seattle--Los Angeles& 10 Gb/s \\
3 & Seattle--Salt Lake City & 10 Gb/s & 4& Salt Lake City--Seatle& 10
Gb/s \\
5 & Los Angeles--Salt Lake City & 10 Gb/s & 6& Salt Lake City--Los
Angeles& 10 Gb/s \\
7& Los Angeles--Houston  & 10 Gb/s & 8& Houston--Los Angeles& 10 Gb/s
\\
9 & Salt Lake City--Kansas City  & 10 Gb/s & 10& Kansas City--Salt
Lake City& 10 Gb/s \\
11 & Kansas City--Houston & 10 Gb/s & 12& Houston--Kansas City& 10
Gb/s \\
13 & Kansas City--Chicago & 20 Gb/s & 14 & Chicago--Kansas City & 20
Gb/s \\
15 & Houston--Atlanta & 10 Gb/s & 16 & Atlanta--Houston & 10 Gb/s \\
17 & Chicago--Atlanta & 10 Gb/s & 18 & Atlanta--Chicago & 10 Gb/s \\
19 & Chicago--New York & 10 Gb/s & 20 & New York--Chicago & 10 Gb/s \\
21 & Chicago--Washington & 10 Gb/s & 22 & Washington--Chicago & 10
Gb/s \\
23 & Atlanta--Washington & 10 Gb/s & 24 & Washington--Atlanta & 10 Gb/s
\\
25 & Washington--New York & 20 Gb/s & 26 & New York--Washington & 20
Gb/s \\
\hline
\end{tabular}
\caption{A description of the 26 links that form the Internet2
   backbone.  Also provided are the id numbers used in this paper for
   notational simplicity.}
\label{tab:I2links}
\end{table}

\begin{figure}
 \begin{center}
 \includegraphics[width=3in]{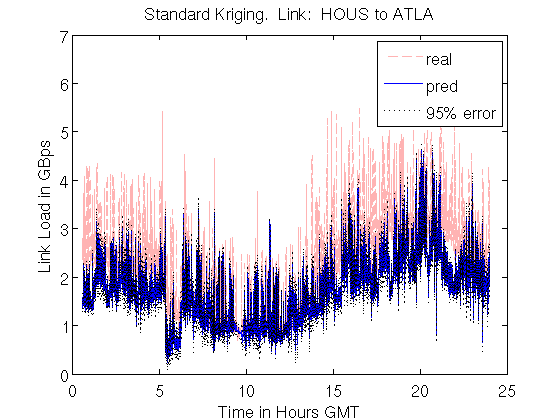}
{\caption{\label{fig:pred_GOOD}
 {\small Prediction for the Internet2 backbone link Houston to Atlanta {\tt HOUS->ATLA} based on the
links: {\tt SEAT->SALT, SEAT->LOSA, LOSA->HOUS, ATLA->WASH, CHIC->NEWY.}
The traces reflect an entire day of activity (February 19, 2009).  
Observe the diurnal patterns and the utilization (see the caption of Fig.\ \ref{fig:I2-topology}).  
The dotted lines indicate $95\%$ prediction bounds. 
}}}
 \vskip -0.2in
\end{center}
\end{figure}

\begin{figure}
\begin{center}
\includegraphics[width=2.5in]{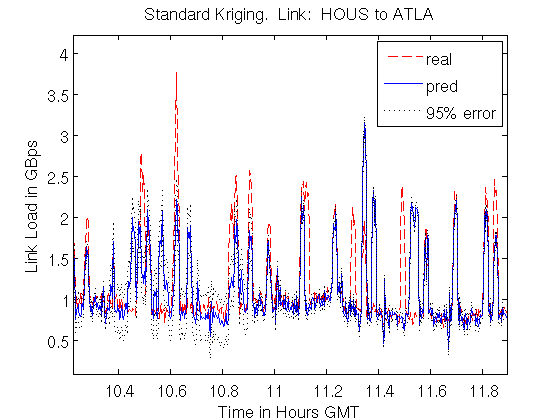}
{\caption{ \label{fig:pred_GOOD_zoom}
 {\small  A zoomed--in portion of Fig.\ \ref{fig:pred_GOOD}.  Observe that the {\em standard kriging estimator}
 closely tracks the true traffic trace and note that the prediction bounds are adequate.  
 }}}
\end{center}
\end{figure}

There are, however, objective limitations to the degree to which one can predict unobserved links from another
set of links.  The example in Fig.\ \ref{fig:pred_BAD} was chosen to illustrate these limitations.
Observe that even though the coarse--scale traffic mean is tracked somewhat, the standard kriging estimator
fails to track the finer scale behavior and the prediction intervals are rather wide.  

\begin{figure}
\begin{center}
\includegraphics[width=3in]{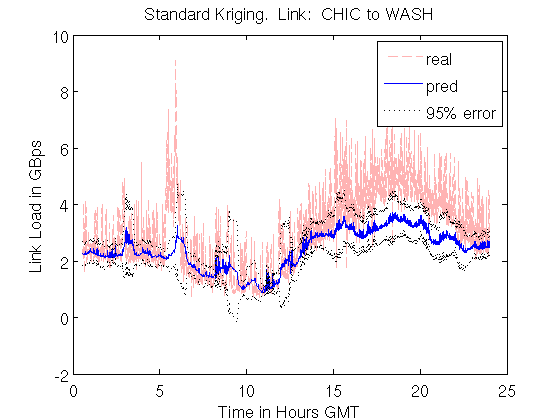}
{\caption{ \label{fig:pred_BAD}
{
\small  Prediction for the Internet2 backbone link Chicago to Washington based on the same
set of observed links as in Fig.\ \ref{fig:pred_GOOD}.  
}}}
 \vskip -0.2in
\end{center}
\end{figure}

\begin{figure}
\begin{center}
\includegraphics[width=3in]{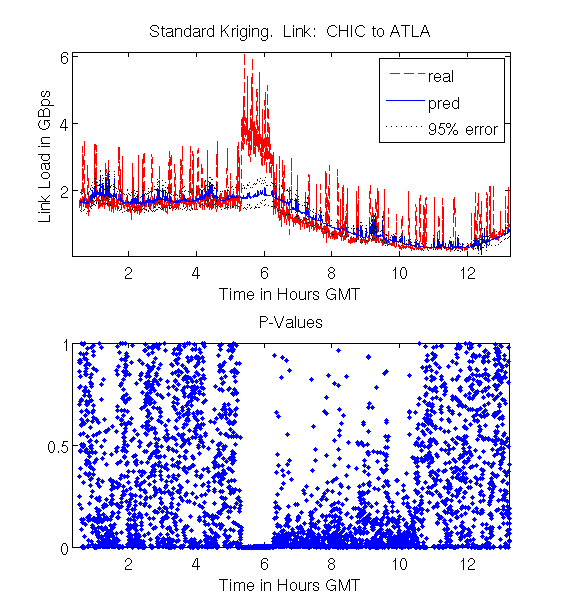}
{\caption{ \label{fig:P_vals}
{
\small  {\em Top panel:} standard kriging for {\tt CHIC->ATLA} via the same
set of observed links as in Fig.\ \ref{fig:pred_GOOD}.  {\em Bottom panel:} 
P--values corresponding to the top panel.
}}}
 \vskip -0.2in
\end{center}
\end{figure}

Fig.\ \ref{fig:P_vals} shows that network kriging may be used in 
anomaly detection as a {\em diagnostic} device.  Namely, one can predict an observed link from a set 
of other observed links and thus obtain a two--sided p--value based on the prediction distribution.  Low p--values
would indicate the presence of an anomaly.  This is well demonstrated in Fig.\ \ref{fig:P_vals} by drop in p--values
between 5am and 6am GMT.  The sudden peak load on the Chicago--Atlanta link is not tracked well by the monitored links
and the underlying NetFlow data used to recover the overall mean and variance structure of the flow--level traffic.
Thus, the network kriging methodology, based on our probabilistic model, provides an novel {\em global view}
on the statistical significance of traffic patterns in the network.

\section{Discussion}
 \label{s:discussion}

In this paper, we developed a probabilistic framework for network--wide modeling of traffic, based 
on multiple sources and large time--scale limit approximations. It is shown that depending on the scaling, 
a  {\em fast} and {\em slow} regime occur in the limit. As an extension, one can also consider
simultaneous limits as the number of sources
$M=M(T)$ and as the time scale $T$ tend to infinity, as well as other complex asymptotic scenarios.

The proposed model proves mathematically tractable, involving few statistical parameters and therefore perfectly suitable
for addressing a number of important questions for network--wide traffic behavior.  As shown, the model can successfully
predict traffic loads on unobserved links (network kriging), employing only a limited set of link measurements, 
provided that some coarse--scale information about the traffic means is available (e.g. through NetFlow data).  

The developed network kriging methodology has further applications to anomaly detection and diagnostics,  as shown in the example
of Fig.\ \ref{fig:P_vals}. Since the model captures the {\em  joint distribution} of all 
links in the network, the multiple testing problem associated with anomaly detection for a large number of links can 
be successfully handled, as well. Further, as illustrated in Fig.\ \ref{fig:pred_BAD} and Table \ref{tab:rmse}, in the presence of
limited resources, it is important to select an ``optimal set" of links for network
monitoring; this model can be used to address this problem in the context of network kriging.

Finally, estimation of the joint distribution of means and covariances of traffic flows, across time and over the network, constitutes 
a challenging, but also important problem for network engineering. Our ongoing work is addressing this problem through flexible,
parsimonious latent variable models that can be estimated in {\em real time} and {\em without} the need for  the availability of NetFlow
data \cite{vaughan:stoev:michailidis:2009}.

\section{Appendix}

\begin{proposition} \label{p:pos-def} The functional $(f,g)\mapsto \phi_{2H}(f,g)$ 
 in \eqref{e:phi} is positive semi--definite {\em if and only if} $0<H\le 1$.
\end{proposition}
\begin{proof} Since $(t,s)\mapsto \phi_{2H}(tf_0,sf_0),\ t,s\in \bbR$ has the form of the auto--covariance of fBm,
then it follows that necessarily $H \in (0,1]$ (see e.g.\ \cite{samorodnitsky:taqqu:1994book}).  It remains to show
that $\phi_\alpha$ is positive definite for all $\alpha:= 2H \in (0,2]$.  

Let $M_\alpha,\ \alpha\in(0,2]$ be an S$\alpha$S random measure with control measure
$\mu$ and define
$$
\Lambda(f) := \int_E f dM_\alpha,\ \ \forall f\in L^\alpha(\mu),
$$
to be the S$\alpha$S integral of the deterministic function $f$ (see e.g.\ Ch.\ 3 in 
\cite{samorodnitsky:taqqu:1994book}).  Notice that
for all 
$x_j \in \bbC,$ and $f_j\in L^\alpha(\mu)$, with $1\le j\le n$, we have
\begin{eqnarray*}
\E {\Big|} \sum_{j=1}^n x_j e^{i \Lambda(f_j)} {\Big|}^2 &=&
\sum_{j,k=1}^n x_j \overline{x}_k \E e^{i \Lambda(f_j-f_k)}\\
&=& \sum_{j,k=1}^n x_j \overline{x}_k e^{ -\|f_j-f_k\|_\alpha^\alpha}.
\end{eqnarray*}
Since the l.h.s.\  of the last expression is always non--negative, so is the r.h.s.
This shows that the function $r_\alpha(f,g):= e^{-\|f-g\|_\alpha^\alpha},$
$f,g\in L^\alpha(\mu)$ is positive definite.

Now, the proof proceeds as the proof of the positive definiteness of the
auto--covariance function of the fractional Brownian motion (see, e.g.\ 
p.\ 106 in \cite{samorodnitsky:taqqu:1994book}).
 Indeed, for all
 $x_j \in \bbC,$ and $f_j\in L^\alpha,\ 0\le j \le n$, and
for all $\epsilon>0$, we have
\begin{eqnarray}\label{e:pos-def-2}
0 &\le&
 \sum_{j,k=0}^n x_j \overline{x}_k e^{-\epsilon\|f_j-f_k\|_\alpha^\alpha}\\
&=&\sum_{j,k=1}^n x_j\overline{x}_ke^{-\epsilon\|f_j-f_k\|_\alpha^\alpha} 
  + \sum_{j=1}^n x_0 \overline{x}_k e^{-\epsilon\|f_0-f_k\|_\alpha^\alpha} \nonumber\\
& & \quad\quad +  \sum_{j=1}^n  x_j \overline{x}_0  e^{-\epsilon\|f_j-f_0\|_\alpha^\alpha} 
    + x_0 \overline{x}_0\nonumber\\
& =:& S_1 + S_2+ S_3 + |x_0|^2 \nonumber
\end{eqnarray}
Since $x_0$ and $f_0$ are at our disposal, let $f_0 := 0$ and
$x_0:= - \sum_{j=1}^n x_j e^{-\epsilon\|f_j\|_\alpha^\alpha}$.  Observe that
with this choice of $x_0$ and $f_0$, we get
$$
 S_2 = S_3 = - |x_0|^2 = -\sum_{j,k=1}^n x_j \overline{x}_k 
e^{-\epsilon\|f_j\|_\alpha^\alpha-\epsilon\|f_k\|_\alpha^\alpha},
$$
and therefore, $S_1 + S_2 + S_3 + |x_0|^2 $ equals:
\begin{eqnarray}\label{e:pos-def-1}
& & \sum_{j,k=1}^n x_j \overline{x}_k 
 {\Big (} e^{-\epsilon\|f_j-f_k\|_\alpha^\alpha} - 
  e^{-\epsilon\|f_j\|_\alpha^\alpha-\epsilon\|f_k\|_\alpha^\alpha} {\Big)}\\
& = & \epsilon \sum_{j,k=1}^n
  x_j \overline{x}_k {\Big(}\|f_j\|_\alpha^\alpha + \|f_k\|_\alpha^\alpha - \|f_j-f_k\|_\alpha^\alpha {\Big)}
   + o(\epsilon),\nonumber
\end{eqnarray}
as $\epsilon \downarrow 0$, where the last relation we used the fact that $e^{-\epsilon a} - e^{-\epsilon b}
 = \epsilon (b-a) + o(\epsilon)$, as $\epsilon \downarrow 0$.  If for some $x_j$'s and $f_j$'s we have
 $\sum_{j,k=1}^n
  x_j \overline{x}_k (\|f_j\|_\alpha^\alpha + \|f_k\|_\alpha^\alpha - \|f_j-f_k\|_\alpha^\alpha ) < 0$,
then, for all sufficiently small $\epsilon>0$, the l.h.s.\  of \refeq{pos-def-1} becomes negative, which
in view of \refeq{pos-def-2}, is impossible.  This shows that $\phi_\alpha$ is positive (semi--)definite.
\end{proof}

\begin{proof}[Proof of Prosition \ref{p:f-fBm}] To check the equality in distribution of two zero mean 
Gaussian processes, it is enough to show the equality of their auto--covariance functions. Let
$1\le \ell_1, \ell_2 \le L$ and $t_1, t_2\ge 0$.  Then, by using the independence of the $B_H^{(j)}(t)$'s, we have:
\begin{eqnarray}\label{e:p:f-fBm-1}
 & & \E (A \vec B_H(t_1))_{\ell_1}(A \vec B_H(t_2))_{\ell_2} \nonumber\\ 
& & \ \ \ \ \ =\sum_{1\le u \le {\cal J}}  r(u)^2 1_{A_{\ell_1} \cap A_{\ell_2}} (u) 
r_{2H}(t_1,t_2),
\end{eqnarray}
where $r_{2H}(t_1,t_2) = (\sigma^2/2)(|t_1|^{2H} + |t_2|^{2H} - |t_1-t_2|^{2H})$. 
On the other hand, as in 
\eqref{e:space-time-cov}, we obtain
\begin{eqnarray*}
& &\E B(t_1 f_{\ell_1}) B(t_2 f_{\ell_2}) = \frac{\sigma^2}{2}
 \sum_{1\le u\le {\cal J}} r(u)^2 {\Big(} |t_1|^{2H} 1_{A_{\ell_1}}(u) \\
& & \ \ \ \ \  + |t_2|^{2H} 1_{A_{\ell_2}}(u) - |t_1 1_{A_{\ell_1}}(u) - t_2 1_{A_{\ell_2}}(u)|^{2H} {\Big)},
\end{eqnarray*}
which after cancellations, equals the r.h.s.\ of \eqref{e:p:f-fBm-1}.
\end{proof}

\begin{proof}[Proof of Proposition \ref{p:properties}]
{\it (i):}  The auto--covariance function of the process $\{B(cf)\}_f$ is
$$
\phi_{2H}(cf,cg) = \frac{\sigma^2}{2}{\Big(} \|cf\|^{2H} + \|cg\|^{2H} - \|cf - cg\|^{2H} {\Big)},
$$
which equals $\E (c^HB(f)) (c^HB(g))$, where $\|\cdot\|$ stands for $\|\cdot\|_{2H}$.
This implies \eqref{e:H-ss}.  The proof of {\it (ii)} is similar.  

For simplicity, let now $\sigma =1$. Then
\begin{eqnarray*}
& & \E (B(f+h) - B(h))(B(g+h) - B(h)) \\
& & \ \  = \frac{1}{2}{\Big(}\|f+h\|^{2H} + \|g+h\|^{2H} - \|f-g\|^{2H} {\Big)} \\
& & \ \ \ \ -\frac{1}{2}{\Big(} \|h\|^{2H} + \|g+h\|^{2H} - \|g\|^{2H} {\Big)} \\
& & \ \ \ \  - \frac{1}{2}{\Big(}\|f+h\|^{2H} +\|h\|^{2H} - \|f\|^{2H} {\Big)}\\
& &\ \ \ \ \ + \|h\|^{2H} = \frac{1}{2} {\Big(}\|f\|^{2H} + \|g\|^{2H} - \|f-g\|^{2H} {\Big)},
\end{eqnarray*}
which equals $\E B(f) B(f)$, and thus implies the equality of the finite--dimensional distributions
in \refeq{stat-inc}.

{\it (iii)} Since $fg = 0$ $\mu-$a.e.\
\begin{eqnarray*}
 \|f-g\|^{2H} & = &\int_{E\cap \{f\not=0\}} |f|^{2H} d\mu + \int_{E \cap \{g\not = 0\}} |g|^{2H} d\mu \\
 & & \ \ \  = \|f\|^{2H} + \|g\|^{2H},
\end{eqnarray*}
where for simplicity, $\sigma =1$. This implies the independence of $B(f)$ and $B(g)$, since
$\E B(f) B(g) = \phi_{2H}(f,g) = 0$ in view of \eqref{e:phi}.

Part {\it (iv)} follows trivially from \eqref{e:phi}. Now, to prove {\it (v)}, it
is enough to show $\E (B(f+g) - B(f) - B(g))^2 =0$ {\em if and only if} $fg =0$ $\mu-$a.e.
It can be shown that the last expectation equals:
\begin{eqnarray*}
& & \E B(f+g)^2 + \E B(f)^2 + \E B(g)^2  - 2\E B(f+g)B(f) \\
& & \quad\quad\quad - 2 \E B(f+g)B(g) + 2 \E B(f) B(g) \\ 
 & & \ \ \  = \|f+g\|^{2H} + \|f\|^{2H} + \|g\|^{2H}  \\
 & & \ \ \ \ - (\|f+g\|^{2H} + \|f\|^{2H} - \|g\|^{2H}) \\ 
 & & \ \ \ \ - (\|f+g\|^{2H} + \|g\|^{2H} - \|f\|^{2H} ) \\
 & & \ \ \ + \|f\|^{2H} +\|g\|^{2H} - \|f-g\|^{2H}  \\
 & & \ \ \ =  2 \|f\|^{2H} + 2 \|g\|^{2H} - \|f-g\|^{2H} - \|f+g\|^{2H}.
\end{eqnarray*}
Since $0<2H<2$, the last expression vanishes if and only if $fg = 0$ $\mu-$a.e. (see Eq.\ (2.7.9) in
Lemma 2.7.14 of \cite{samorodnitsky:taqqu:1994book}).
\end{proof}

\begin{proof}[Proof of Proposition \ref{p:properties-fLsm}] Let as in \cite{samorodnitsky:taqqu:1994book},
$\|\xi\|_\alpha$ denote the scale coefficient of the $\alpha$ stable random variable $\xi$.  To prove {\it(i)},
it suffices to show that for all $f_j\in L^1(\mu)$, and $\theta_j\in\bbR,\ 1\le j \le n$, we have that
$$
\| \sum_{1\le j\le n} \theta_j \Lambda(cf_j)\|_\alpha^\alpha =  
 \| c^{1/\alpha} \sum_{1\le j\le n} \theta_j \Lambda(f_j)\|_\alpha^\alpha.
$$
The l.h.s.\ of this expression equals:
\begin{eqnarray*}
\int_{\bbR \times E} {\Big|}\sum_{1\le j\le n} \theta_j (1_{(-\infty,cf_j(u))}(x) 
  - 1_{(-\infty,0)}(x)){\Big|}^\alpha dx \mu(du).
\end{eqnarray*}
By setting $z:= x/c$, we obtain that the last integral equals:
$$
c \int_{\bbR \times E} {\Big|}\sum_{1\le j\le n} \theta_j (1_{(-\infty,f_j(u))}(z) - 
 1_{(-\infty,0)}(z)){\Big|}^\alpha dz \mu(du),
$$
which is 
$$ 
 \| c^{1/\alpha} \sum_{1\le j\le n} \theta_j \Lambda(f_j)\|_\alpha^\alpha.
$$
This completes the proof of {\it (i)}.

Part {\it (ii)} can be established similarly by using the Fubini's theorem and the change of variables
$z:=x-h(u)$ in the integral
$$
 \int_{\bbR\times E} {\Big|} \sum_{1\le j\le n} \theta_j (1_{(-\infty,f(u)+h(u))}(x) - 
1_{(-\infty,h(u))}(x)){\Big|}^\alpha dx \mu(du),
$$
which equals $\|\sum_{1\le j\le n} \theta_j (\Lambda(f+h)-\Lambda(h))\|_\alpha^\alpha$.
 
{\it (iii):} In view of Theorem 3.5.3 in \cite{samorodnitsky:taqqu:1994book}, $\Lambda(f)$ and $\Lambda(g)$ are independent
{\em if and only if} 
$$
 (1_{(-\infty,f(u))}(x) - 1_{(-\infty,0)}(x)) (1_{(-\infty,g(u))}(x) - 1_{(-\infty,0)}(x)) =0,
$$
for $dx\times\mu(du)$ almost all $(x,u)$.  By considering cases for the signs of  $f(u)$ and $g(u)$, it follows that
the latter equality holds (for $dx\times\mu(du)$ almost all $(x,u)$) if and only if $f(u)g(u)\le 0$ $\mu(du)$--a.e.

{\it (iv):} Let $f_0:= 0$ and observe that 
$$
\Lambda(f_k) -\Lambda(f_{k-1}) = \int_{\bbR\times E} 1_{A_k}(x,u) M_\alpha(dx,du), 
$$
where $A_k = \{ (x,u): f_{k-1}(u)\le x< f_{k}(u)\}$, for $1\le k\le n$.  Again, by Theorem 3.5.3
in \cite{samorodnitsky:taqqu:1994book}, we have that the above increments are independent, if and only if,
the sets $A_k,\ 1\le k\le n$ are mutually disjoint (mod $dx\times d\mu$), which is clearly the case here.

The proof of {\it (v)} is similar to that of {\it (iv)}. 

{\it (vi):} Follows {\it (ii)} and part {\it (iv)} applied to the independent processes 
$\Lambda (t f_+)$ and $\Lambda(t f_-)$, where $f_\pm = \max\{\pm f,0\}$ since 
$\Lambda(tf) = \Lambda(tf_+) - \Lambda(tf_-).$
\end{proof}

\begin{proof}[Proof of Proposition \ref{p:f-fBm-int-rep}]
To prove that $\E B(f) B(g) = \phi_{2H}(f,g)$, it suffices to show that
$$
{\rm Var}(B(f)-B(g)) = \E (B(f)-B(g))^2 = \sigma^2 \|f-g\|_{2H}^{2H}.
$$
This is indeed the case: By using changes of variables and Fubini's theorem, we have that
\begin{eqnarray*}
& & \E (B(f)-B(g))^2 \\
& & \ \ \ = \int_{\bbR\times E}  {\Big(}(f(u)-x)_+^{H-1/2} - (g(u)-x)_+^{H-1/2}{\Big)}^2dx\mu(du),
\end{eqnarray*}
equals
\begin{eqnarray*}
& &\int_{\bbR\times E} |f(u)-g(u)|^{2H-1} {\Big(}(1-x/c(u))_+^{H-1/2} \\
& & \ \ \ \ - (-x/c(u))_+^{H-1/2}{\Big)}^2dx\mu(du)\\
& & \ = \sigma^2 \int_{E} |f(u)-g(u)|^{2H} \mu(du) = \sigma^2 \| f-g\|_{2H}^{2H},
\end{eqnarray*}
where $c(u) = f(u)-g(u)$ and $\sigma^2 = \int_{\bbR} ((1-x)_+^{H-1/2} - (-x)_+^{H-1/2})^2 dx$.
\end{proof}

%

\smallskip
\begin{proof}[Proof of Proposition \ref{p:std-Kriging}]
Let $\mu_Y = \E Y(t_0) = (\mu_u^t,\ \mu_o^t)^t$ and
$$
 \Sigma_Y = \E (Y(t_0) -\mu_Y) (Y(t_0) -\mu_Y)^t = \left(\begin{array}{ll}
 \Sigma_{uu} & \Sigma_{uo}\\
 \Sigma_{ou} & \Sigma_{oo} \end{array} \right),
$$
where $\Sigma_{ij} = A_i \Sigma_X A_j^t,\ i,j\in\{o,u\}$.  
The conditional distribution of $Y_u(t_0)| Y_o(t_0)$ is Gaussian and:
\begin{equation}
 Y_u(t_0) | Y_o(t_0) \sim {\cal N} (\mu_u+\Sigma_{uo}\Sigma^{-1}_{oo}(Y_o-\mu_o),
 \Sigma_{uu}-\Sigma_{uo}\Sigma^{-1}_{oo}\Sigma_{ou})
 \label{eqn.cdist} 
\end{equation}
(see eg Theorem 1.6.6 in \cite{brockwell:davis:1991}). Thus, an unbiased predictor of $Y_u(t_0)$, given $Y_o(t_0)$ is:
\begin{equation}
 \what Y_u(t_0):= \E ( Y_u(t) | Y_o(t) )= \mu_u + \Sigma_{uo}\Sigma^{-1}_{oo}(Y_o-\mu_o).
\label{eqn.condmse}
\end{equation}
 This implies \eqref{e:p:std-Kriging-Y-hat} 
and \eqref{e:p:std-Kriging-MSE} with ${\cal D}$ replaced by $Y_o(t_0)$. Proposition \ref{p:space-time-factor}
below implies, however, that $\what Y_u(t_0) - Y_u(t_0)$ and $Y_o(t)$ are
uncorrelated, for all $t_0 -m \le t \le t_0$. This completes the proof of {\it (i)}.

\medskip
\noi To prove {\it (ii)}, let $\theta$ be a constant vector of the same dimension
as $Y_u(t_0)$. Consider the random variable $\xi:= \theta^t Y_u(t_0)$.  It is well--known
that $\E (\xi | Y_o(t_0))$ is the best unbiased m.s.e.\ predictor of $\xi$ via $Y_o(t_0)$.  Thus
$$
 \theta^t \MSE(\what Y_u(t_0)| Y_o(t_0)) \theta 
 \le \theta^t \MSE(Y_u^*(t_0)| Y_o(t_0)) \theta,
$$
which implies {\it (ii)} and completes the proof.
\end{proof}

%
%

\noi The following result shows that if the space--time covariance structure of a random field factors, then
 the instantaneous standard kriging estimate is an optimal linear predictor even in the presence of
 additional data from the past.

\begin{proposition}\label{p:space-time-factor}
 Let $\{\xi(t,x)\}_{(t,x) \in T\times S}$ be a finite variance space--time random field.  Suppose that
$\E \xi(t,x) = 0$, for all $(t,x)\in T\times S$ and that
$$
{\rm Cov }(\xi(t,x),\xi(s,y)) = \gamma(t,s) R(x,y),
$$
for all $t,s\in T$ and $x,y\in S.$ Consider the data set ${\cal D} = \{ \xi(t_i,x_j),\ 0\le i \le m,\ 1\le j\le n\}$ 
of observations of the random field at times 
$t_0, t_1,\cdots, t_m$ and locations $x_1,\cdots x_n$. Then, there exist coefficients $\beta_j, 1\le j\le n$, such that
\begin{equation}\label{e:p:space-time-factor}
 \what \xi (t_0, x_0) := \sum_{j=1}^n \beta_j \xi(t_0,x_j)
\end{equation}
is the best linear in ${\cal D}$, unbiased predictor of $\xi(t_0,x_0)$.  In particular, we have
\begin{equation}\label{e:p:space-time-factor-beta}
 \vec \beta = \Sigma_{t_0}^{-} \vec c,\ \ \mbox{ where }\ \ 
  \Sigma_{t_0} = ({\rm Cov}(\xi(t_0,x_i),\xi(t_0,x_j)))_{n\times n}
\end{equation}
and $\vec c = ({\rm Cov}(\xi(t_0,x_0),\xi(t_0,x_i)))_{i=1}^n$.  Here $\Sigma_{t_0}^{-}$ denotes the Moore--Penrose
 generalized inverse
of the covariance matrix $\Sigma_{t_0}$.
\end{proposition}

\begin{proof} Consider the Hilbert space ${\cal L}^2$ of finite variance random variables with zero means and the usual 
inner product $\langle \xi,\eta\rangle :=\E \xi\eta$.  Consider the sub--space  $W = {\rm span}({\cal D}) \le {\cal L}^2$
and observe that the best linear in ${\cal D}$ unbiased predictor for $\xi(t_0,x_0)$ is the (unique) orthogonal projection
of $\xi(t_0,x_0)$ onto $W$.

Let $\what \xi(t_0,x_0)$ be the orthogonal projection of $\xi(t_0,x_0)$ onto the smaller subspace
${\rm span}\{\xi(t_0,x_j),\ 1\le j\le n\}$.  We then have that, for all $k =1,\cdots,n$
\begin{eqnarray*}
0 &=& {\rm Cov} {\Big(} (\xi(t_0,x_0) -  \sum_{j=1}^n \beta_j \xi(t_0,x_j)),\  \xi(t_0,x_k) {\Big)}\\
  &=& \gamma(t_0,t_0) R(x_0,x_k) - \sum_{j=1}^{n}\beta_j \gamma(t_0,t_0) R(x_j,x_k).
\end{eqnarray*}
This, since $\gamma(t_0,t_0)\not = 0$, shows that
\begin{equation}\label{e:p:space-time-factor-1}
  R(x_0,x_k) - \sum_{j=1}^{n}\beta_j R(x_j,x_k) = 0,\ \ \mbox{ for all } 1\le k \le n.
\end{equation}
We will show next that $\xi(t_0,x_0) - \sum_{j=1}^n \beta_j \xi(t_0,x_j)$ is orthogonal 
to $\xi(t_i,x_k)$ for all $i=1,\cdots,m$ and 
$k=1,\cdots,n$. Indeed,
\begin{eqnarray*}
& & {\rm Cov}{\Big(}\xi(t_0,x_0) - \sum_{j=1}^n \beta_j \xi(t_0,x_j), \xi(t_i,x_k) {\Big)}\\
 & &\ \ \ = \gamma(t_0,t_i) R(x_0,x_k) - \sum_{j=1}^n \beta_j \gamma(t_0,t_i) R(x_j,x_k) \\
 & &\ \ \ = \gamma(t_0,t_i)  {\Big(} R(x_0,x_k) - \sum_{j=1}^n \beta_j R(x_j,x_k) {\Big)} = 0,
\end{eqnarray*}
where the last term vanishes because of \eqref{e:p:space-time-factor-1}.  This implies that $\what \xi(t_0,x_0)$
is in fact the orthogonal projection of $\xi(t_0,x_0)$ onto $W$ and hence, it is the b.l.u.p.\ in terms of the data
in ${\cal D}$.

Relation \eqref{e:p:space-time-factor-beta} follows by solving \eqref{e:p:space-time-factor-1}.  
If $\Sigma_{t_0}$ is invertible, then the solution is certainly unique, otherwise the Moore--Penrose 
generalized inverse $\Sigma_{t_0}^{-}$ yields a particular natural solution.
\end{proof}

\medskip
\begin{proof}{(Proposition \ref{p:temporal-pred})}
  Part {\it (i)} is standard in one dimension (see eg Corollary 5.1.1 in \cite{brockwell:davis:1991}). 
For completeness, we will prove the result in the case when $Y_o(t)\in\bbR^d$. Let
$\Sigma_{oo} = A_o \Sigma_X A_o^t$ and observe that $\E (Y_o(t) -\mu_o)(Y_o(s)-\mu_o)^t 
= \gamma_X(|t-s|) \Sigma_{oo}$.

Consider now the zero mean Gaussian vectors: $\xi := Y_o(t_0+h) - \mu_o$ and
$$\eta  = ( Y_o(t_0)^t-\mu_o^t,\ \cdots,\ Y_o(t_0-m)^t-\mu_o^t)^t.
$$
Note that $\xi \sim {\cal N}(0,\Sigma_{oo})$ and $\eta \sim {\cal N}(0,\Gamma_{m+1}\otimes \Sigma_{oo})$, 
where '$\otimes$' denotes the Kronecker product:
$$
\Gamma_{m+1} \otimes \Sigma_{oo} = {\Big(} \gamma_X(|i-j|) \Sigma_{oo} {\Big)}_{(m+1)d \times (m+1)d},
$$
and where $\Sigma_{oo}$ is a $d\times d$ matrix.  By assumption, we have that $\Sigma_{oo}$ is invertible,
and as argued above, so is the Toeplitz matrix $\Gamma_{m+1}$, since $\gamma_X(k)\to 0,\ k\to\infty$
(Proposition 5.1.1 in \cite{brockwell:davis:1991}).  This implies that 
$\Sigma_{\eta\eta}^{-1} :=
(\Gamma_{m+1} \otimes \Sigma_{oo})^{-1} = \Gamma_{m+1}^{-1} \otimes \Sigma_{oo}^{-1}$ exists.
Therefore, the conditional distribution $\xi|\eta$ is as follows:
\begin{equation}\label{e:xi|eta}
 \xi|\eta \sim {\cal N} {\Big(} \Sigma_{\xi\eta} \Sigma_{\eta\eta}^{-1} \eta,\ \Sigma_{\xi\xi } -
  \Sigma_{\xi\eta} \Sigma_{\eta\eta}^{-1}
   \Sigma_{\eta\xi} {\Big)},
\end{equation}
where 
$$
\Sigma_{\xi\eta} = \E \xi \eta^t = \vec \gamma_{m+1}(h)^t \otimes \Sigma_{oo},\ \ \ 
\Sigma_{\eta\eta} = \E \eta \eta^t = \Gamma_{m+1}\otimes \Sigma_{oo},
$$
and
$$
 \Sigma_{\eta\xi}
 = \vec \gamma_{m+1}(h) \otimes \Sigma_{oo}.$$
By recalling the definitions of $\xi$ and $\eta$, we obtain that 
\begin{eqnarray*}
& &\E (Y_{o}(t_0+h) | {\cal D}) = \mu_o + \E (\xi|\eta) \\
& &  \ \ \ \  = \mu_o + (\vec \gamma_{m+1}(h)^t \otimes \Sigma_{oo})
(\Gamma_{m+1} ^{-1} \otimes \Sigma_{oo}^{-1}) \eta\\
& &\ \ \ \  =\mu_o + (\vec \gamma_{m+1}(h)^t \Gamma_{m+1}^{-1})\otimes 
 {\bbI}_d) \eta,
\end{eqnarray*}
which equals \eqref{e:temporal-pred-i}, and where in the last relation
we used the mixed--product property of the Kronecker product. By Relation \eqref{e:xi|eta}, we also have
\begin{eqnarray*}
& & \MSE(\what Y_o(t_0+h)|{\cal D}) = \Sigma_{\xi\xi} -\Sigma_{\xi\eta} \Sigma_{\eta\eta}^{-1} \Sigma_{\xi\eta}
 = \gamma_X(0) \Sigma_{oo} \\
& & \ \ -
  (\vec \gamma_{m+1}(h)^t \otimes \Sigma_{oo}) (\Gamma_{m+1}^{-1}\otimes \Sigma_{oo}^{-1})
(\vec \gamma_{m+1}(h) \otimes \Sigma_{oo}) \\
 &  & = \gamma_X(0) \Sigma_{oo} \\
 & & \ \ - (\vec \gamma_{m+1}(h)^t \Gamma_{m+1}^{-1} \gamma_{m+1}(h)) \otimes \Sigma_{oo} \equiv \sigma^2(h) \Sigma_{oo},
\end{eqnarray*}
by the mixed--product property of the Kronecker product and the fact that
$\gamma_{m+1}(h)^t \Gamma_{m+1}^{-1} \gamma_{m+1}(h)$ is a scalar. We have thus shown
 \eqref{e:temporal-pred-i-mse}.

\medskip
We now focus on proving {\it (ii)}. Consider $Y_{o}(t_0+h)$ and write
$$
 \what Y_u(t_0+h) := \mu_u + C (Y_o(t_0+h) - \mu_o) + C (\what Y_o(t_0+h) - Y_o(t_0+h)).
$$
As in Proposition \ref{p:std-Kriging}, one can show that  $Y_u(t_0+h) - C Y_o(t_0+h)$ is
independent from $Y_o(t)$, for all $t\le t_0+h$.  Therefore, 
\begin{eqnarray*}
\MSE(\what Y_u(t_0+h)|{\cal D}) &=&  \MSE(\what Y_u(t_0+h)|Y_o(t_0+h)) \\
 & & \ \ \ + C \MSE(\what  Y_o(t_0+h) | {\cal D}) C^t,
\end{eqnarray*}
where in the last relation $ \MSE(\what Y_u(t_0+h)|Y_o(t_0+h)) $ stands for the m.s.e.\ of the
standard Kriging estimator in Relation \eqref{e:p:std-Kriging-MSE} and where 
$\MSE(\what  Y_o(t_0+h) | {\cal D})$ is as in \eqref{e:temporal-pred-i-mse}.  This completes the proof 
of {\it (ii)}.

\medskip
To prove {\it (iii)} observe that the estimator in {\it (i)} is the conditional expectation 
of $Y_o(t_0+h)$ given ${\cal D}$ and it is therefore the best m.s.e.\ predictor. If $Y(t)$ is non--Gaussian,
this yields only the b.l.u.p.  By Proposition \ref{p:std-Kriging}, we have that
$$
 \E {\Big(} Y_u(t_0+h)\, |\, \{ Y_o(t),\ t\le t_0+h\} {\Big)} = \mu_u + C (Y_o(t_0+h) - \mu_o),
$$
on the other hand, by part {\it (i)}, we have that 
$$
\E {\Big(} \mu_u + C (Y_o(t_0+h) - \mu_o) | {\cal D} {\Big)} = \mu_u + C (\what Y_o(t_0+h) - \mu_o).
$$
The last two relations yield:
$
\E {\Big(} Y_u(t_0+h) | {\cal D} {\Big)} = \mu_u + C (\what Y_o(t_0+h) - \mu_o),
$
which shows that $\what Y_u(t_0+h)$ is the best m.s.e.\ predictor.  In the non--Gaussian case, this is
merely the b.l.u.p.
\end{proof}




\bibliographystyle{plain}


\end{document}